\documentclass{tlp}

\usepackage[utf8]{inputenc}
\usepackage{amsmath}
\usepackage{amssymb}
\usepackage{microtype}
\usepackage{url}
\usepackage{xcolor}
\usepackage{longtable}
\usepackage{ifthen}
\usepackage{multicol}
\usepackage{float}
\usepackage{subcaption}
\usepackage{graphicx}
\usepackage{cancel}

\newcommand{\tuple}[1]{\ensuremath{\langle #1 \rangle}}
\newcommand{\Htrace}{\ensuremath{\mathbf{H}}}
\newcommand{\Ttrace}{\ensuremath{\mathbf{T}}}
\newcommand{\M}{\ensuremath{\mathbf{M}}}

 \RequirePackage{bm}
\RequirePackage{textcomp}
\RequirePackage{upgreek}

\RequirePackage{xcolor}
\makeatletter

\input pdf-trans
\newbox\qbox
\def\usecolor#1{\csname\string\color@#1\endcsname\space}
\newcommand\bordercolor[1]{\colsplit{1}{#1}}
\newcommand\fillcolor[1]{\colsplit{0}{#1}}
\newcommand\outline[1]{\leavevmode%
  \def\maltext{#1}%
  \setbox\qbox=\hbox{\maltext}%
  \boxgs{Q q 2 Tr \thickness\space w \fillcol\space \bordercol\space}{}%
  \copy\qbox%
}
\makeatother
\newcommand\colsplit[2]{\colorlet{tmpcolor}{#2}\edef\tmp{\usecolor{tmpcolor}}%
  \def\tmpB{}\expandafter\colsplithelp\tmp\relax%
  \ifnum0=#1\relax\edef\fillcol{\tmpB}\else\edef\bordercol{\tmpC}\fi}
\def\colsplithelp#1#2 #3\relax{%
  \edef\tmpB{\tmpB#1#2 }%
  \ifnum `#1>`9\relax\def\tmpC{#3}\else\colsplithelp#3\relax\fi
}
\bordercolor{black}
\fillcolor{white}
\def\thickness{.3}

\newcommand{\next}{\text{\rm \raisebox{-.5pt}{\Large\textopenbullet}}}  \newcommand{\previous}{\text{\rm \raisebox{-.5pt}{\Large\textbullet}}}  \newcommand{\wnext}{\ensuremath{\widehat{\next}}}
\newcommand{\wprevious}{\ensuremath{\widehat{\previous}}}
\newcommand{\alwaysF}{\ensuremath{\square}}
\newcommand{\alwaysP}{\ensuremath{\blacksquare}}
\newcommand{\eventuallyF}{\ensuremath{\Diamond}}
\newcommand{\eventuallyP}{\ensuremath{\blacklozenge}}
\newcommand{\until}{\ensuremath{\mathbin{\mbox{\outline{$\bm{\mathsf{U}}$}}}}}
\newcommand{\release}{\ensuremath{\mathbin{\mbox{\outline{$\bm{\mathsf{R}}$}}}}}
\newcommand{\since}{\ensuremath{\mathbin{\bm{\mathsf{S}}}}}
\newcommand{\trigger}{\ensuremath{\mathbin{\bm{\mathsf{T}}}}}
\newcommand{\finally}{\ensuremath{\mbox{\outline{$\bm{\mathsf{F}}$}}}}
\newcommand{\initially}{\ensuremath{\bm{\mathsf{I}}}}

\mathchardef\mhyphen="2D

\newcommand{\intervcc}[2]{\ensuremath{[#1..#2]}}
\newcommand{\intervco}[2]{\ensuremath{[#1..#2)}}
\newcommand{\intervoc}[2]{\ensuremath{(#1..#2]}}
\newcommand{\intervoo}[2]{\ensuremath{(#1..#2)}}
\newcommand{\rangecc}[3]{\ensuremath{#1 \in \intervcc{#2}{#3}}}
\newcommand{\rangeco}[3]{\ensuremath{#1 \in \intervco{#2}{#3}}}
\newcommand{\rangeoc}[3]{\ensuremath{#1 \in \intervoc{#2}{#3}}}

\newcommand{\cl}[1]{\ensuremath{\mathit{cl}({#1})}}

 \providecommand{\logfont}{\textrm}

\newcommand{\HT}{\ensuremath{\logfont{HT}}}

\newcommand{\LTL}{\ensuremath{\logfont{LTL}}}

\newcommand{\THT}{\ensuremath{\logfont{THT}}}

\newcommand{\THTo}{\ensuremath{\THT_{\!\omega}}}
\newcommand{\TEL}{\ensuremath{\logfont{TEL}}}

\newcommand{\DEL}{\ensuremath{\logfont{DEL}}}

\newcommand{\DL}{\ensuremath{\logfont{DL}}}

\newcommand{\MTL}{\ensuremath{\logfont{MTL}}}

\newcommand{\MHT}{\ensuremath{\logfont{MHT}}}
\newcommand{\MHTf}{\ensuremath{\MHT_{\!f}}}
\newcommand{\MHTo}{\ensuremath{\MHT_{\!\omega}}}
\newcommand{\MEL}{\ensuremath{\logfont{MEL}}}
\newcommand{\MELf}{\ensuremath{\MEL_{\!f}}}
\newcommand{\MELo}{\ensuremath{\MEL_{\omega}}}

 \providecommand{\sysfont}{\textit}

\newcommand{\clingo}{\sysfont{clingo}}

\newcommand{\telingo}{\sysfont{telingo}}

\newcommand{\eqdef}{\ensuremath{\mathbin{\raisebox{-1pt}[-3pt][0pt]{$\stackrel{\mathit{def}}{=}$}}}}

\newcommand{\myatom}{\ensuremath{p}}

\newcommand{\PV}{\ensuremath{\mathcal{A}}}

\def\cI{I}
\def\cJ{J}

\newcommand{\metricI}[1]{\ensuremath{#1_{\cI}}}
\newcommand{\metric}[3]{\ensuremath{#1_{
\ifthenelse{\equal{#2}{#3}}
{#2}
{
\ifthenelse{\equal{#2}{0}}
{
\ifthenelse{\equal{#3}{\omega}}
{}
{\leq#3}
				}
{
\ifthenelse{\equal{#3}{\omega}}
{\geq#2}
{\intervco{#2}{#3}}
				}
}}}}

\newcommand{\tmf}{\ensuremath{\tau}}

\newenvironment{proofof}[1]{\noindent {\bf Proof of #1.}}{\bigskip}

\newcommand{\metricol}[3]{\ensuremath{#1_{
\ifthenelse{\equal{#2}{#3}}
{#2}
{
\ifthenelse{\equal{#2}{0}}
{
\ifthenelse{\equal{#3}{\omega}}
{}
{\leq#3}
      }
{
\ifthenelse{\equal{#3}{\omega}}
{\geq#2}
{\intervo{#2}{#3}}
      }
    }}}}

\newcommand{\trivaluation}[3]{\ensuremath{\bm{#3}(#1,#2)}}
\newcommand{\trival}[2]{\trivaluation{#1}{#2}{m}}
\newcommand{\trivalp}[2]{{\bm{m'}}(#1,#2)}

\newcommand{\Lab}[1]{\ensuremath{{\ell_{#1}}}}

\newcommand{\dist}{\ensuremath{\mathit{u}}} \renewcommand{\H}{\ensuremath{\Htrace}}
\newcommand{\T}{\ensuremath{\Ttrace}}

\newcommand{\press}{\mathit{p}}
\newcommand{\moving}{\mathit{m}}
\newcommand{\arrives}{\mathit{a}}
\newcommand{\myex}{\press \to \moving \metric{\until}{0}{3}\arrives}

\newcommand{\mytag}[1]{\ensuremath{\;\mathbf{#1}}}

\def\squareforqed{\hbox{\rlap{$\sqcap$}$\sqcup$}}
\def\qed{\ifmmode\squareforqed\else{\unskip\nobreak\hfil
\penalty50\hskip1em\null\nobreak\hfil\squareforqed
\parfillskip=0pt\finalhyphendemerits=0\endgraf}\fi}

\lefttitle{The authors.}

\jnlPage{1}{8}
\jnlDoiYr{2021}
\doival{10.1017/xxxxx}

\newtheorem{definition}{Definition}
\newtheorem{theorem}{Theorem}
\newtheorem{proposition}{Proposition}
\newtheorem{corollary}{Corollary}
\newtheorem{lemma}{Lemma}

\makeatletter
\begin{document}
\title[Reducing Arbitrary Metric Temporal Formulas into Logic Programs]{Reducing Arbitrary Metric Temporal Formulas into Logic Programs under Answer Set Semantics}

\begin{authgrp}
  \author{\sn{Mart\'{\i}n} \gn{Di\'eguez}}
  \affiliation{University of Angers, France}
  \author{\sn{Susana} \gn{Hahn}}
  \affiliation{University of Potsdam, Germany}\affiliation{Potassco Solutions, Germany}
  \author{\sn{Torsten} \gn{Schaub}}
  \affiliation{University of Potsdam, Germany}\affiliation{Potassco Solutions, Germany}
  \author{\sn{Igor} \gn{St\'ephan}}
  \affiliation{University of Angers, France}
\end{authgrp}

\maketitle

\begin{abstract}
  Metric temporal equilibrium logic (\MEL) extends temporal equilibrium logic (\TEL)
  by incorporating quantitative timing constraints,
  enabling the specification and analysis of deadlines and durations.
  \MEL\ is particularly suited for domains where time-bound properties are crucial,
  such as embedded systems, cyber-physical systems, and real-time software.
  It facilitates the precise expression of timing behaviors,
  such as the requirement that an event must occur within 5 milliseconds of a trigger,
  which often elude traditional qualitative temporal logics.

  In this paper,
  we present a Tseitin-like translation that maps any metric temporal formula into
  a logic programming fragment restricted to past operators.
  This translation provides a formal bridge to leverage existing Answer Set Programming (ASP) solvers
  for reasoning about metric temporal constraints.
  By restricting the target fragment to past operators,
  we enable more effective evaluation and integration with current ASP-based toolchains for multi-shot solving.
\end{abstract}
 \section{Introduction}\label{sec:introduction}

Reasoning about actions and change, or more generally about dynamic systems,
is not only central to knowledge representation and reasoning
but at the heart of computer science~\citep{TIMEHandbook}.
In practice, this kind of reasoning often requires both qualitative as well as quantitative dynamic constraints.
For instance, when planning and scheduling at once,
actions may have durations, and their effects may need to meet deadlines.
On the other hand, any flexible formalism for actions and change must incorporate some form of non-monotonic reasoning to deal with inertia and other types of defaults.

Over the past years, we addressed qualitative dynamic constraints by combining traditional approaches,
like Dynamic and Linear Temporal Logic (\DL~\citep{hatiko00a} and \LTL~\citep{pnueli77a}),
with the base logic of Answer Set Programming (ASP~\citep{lifschitz99b}), namely,
the logic of Here-and-There (\HT~\citep{heyting30a}) and its non-monotonic extension, called Equilibrium Logic~\citep{pearce96a}. This resulted in (the non-monotonic formalisms) linear dynamic and temporal equilibrium logics
(\DEL~\citep{bocadisc18a} and \TEL~\citep{agcadipescscvi20a})
that gave rise to the temporal ASP system \telingo~\citep{cakamosc19a} extending the ASP system
\clingo~\citep{gekakaosscwa16a}.

A commonality of such dynamic and temporal logics is that
they abstract from specific time points when capturing temporal relationships.
For instance, in temporal logic, we can use the formula
\(
\alwaysF ( \mathit{press} \to \mathit{moving}\until \mathit{arrives})
\)
to express that always after pressing the button, an elevator is moving until it arrives at the corresponding floor.
However,
nothing can be said about how long the elevator is moving until it arrives.

A key design decision was to base both aforementioned logics, \TEL\ and \DEL, on the same linear-time semantics.
\cite{becadiscsc24a} introduce a variant of \TEL\ that
associates each state with a discrete time point.
This variant, called \emph{Metric Temporal Equilibrium Logic} (\MEL),
is defined in terms of a monotonic logic called  \emph{Metric Temporal logic of Here-and-There} (\MHT),
a metric temporal extension of the intermediate logic of Here-and-There, mentioned above,
plus a criterion for selecting minimal \MHT\ models called (metric) equilibrium models.
We obtain in this way a non-monotonic entailment relation.
For instance, in our example, we may thus express that whenever the button is pressed,
the elevator will arrive at its destination  in $0$ to $3$ seconds, by writing
\begin{equation}
  \alwaysF (\mathit{press} \to \mathit{moving} \metric{\until}{0}{3}\mathit{arrives}). \label{ex:press:moving}
\end{equation}
Unlike the non-metric version,
this stipulates that
once $\mathit{press}$ is true in a state,
$\mathit{arrives}$ must be true in some future state
whose associated time is at most $3$ time units after the time of $\mathit{press}$, and $\mathit{moving}$ is true in all states in between.

An important step towards an implementation consists in
reducing arbitrary input theories to normal forms \footnote{We use the term \emph{normal form} here because many different subclasses of \MEL\ can be reduced to this form and, in addition, it enables the design of more efficient solving techniques focused solely on theories that satisfy our logic programming format. In propositional logic, conjunctive and disjunctive normal forms are the most common ones in classical logic. In ASP, however, the reduction of arbitrary theories results in disjunctive logic programs~\citep{pearce06a}.} closer to logic programs.
For instance,
\cite{becadiharosc24a,becadirohasc25a} propose a subclass of \MEL\ that can be encoded in propositional ASP.
This translation does not consider binary metric temporal operators such as \emph{metric until} or \emph{metric release}.
Moreover,
the resulting logic program may contain pure future atoms in the rule heads
and that leads to the introduction of future dependencies.
It would be interesting to avoid such future dependendencies since,
so and incremental solvers are more efficient when derivations of the current state can be directly done from information
obtained at previous states (and already computed by the solver) rather than depending on ``future'' states that have not 
been still considered~\citep{gabbay87a}. 

In this paper,
we introduce a Tseitin-like reduction from any arbitrary metric temporal formulas into
metric temporal logic programs.
Our translation generates a specific type of metric temporal logic program in which future operators affect only to implications. However, both the body and the head of each implication contain only past and present references.
The advantage of this format lies in the use of incremental solving techniques already implemented in solvers such as \telingo{}~\citep{cakamosc19a}. In incremental solving, the use of past temporal operators in the body of the rules implies that these bodies can be readily evaluated at each iteration, since their truth values have already been determined in previous steps.
Our translation requires that the intervals associated with binary temporal operators in the input formulas be of the form $\intervcc{0}{n}$, where $n < \omega$.

The rest of the paper is organized as follows.
In Section~\ref{sec:approach}, we present the syntax and semantics of \MHT, the monotonic basis of \MEL.
In Section~\ref{sec:three-valued}, we present an alternative (but equivalent) three-valued semantics
that is used to prove the correctness of our Tseitin-like translation,
which is itself presented in Section~\ref{sec:translation}.
We finish the paper with the conclusions and future work.

 \section{Metric Logic of Here-and-There}\label{sec:approach}

We begin by introducing the logic of Metric Here-and-There (\MHT),
a metric extension of \HT\ which serves as the monotonic foundation for Metric Equilibrium Logic (\MEL).

Given $m \in \mathbb{N}$ and $n \in \mathbb{N} \cup \{\omega\}$, we let
\intervcc{m}{n} stand for the set $\{i \in \mathbb{N} \mid m \leq i \leq n\}$,
\intervco{m}{n}       for         $\{i \in \mathbb{N} \mid m \leq i < n\}$,
\intervoc{m}{n}       for         $\{i \in \mathbb{N} \mid m < i \leq n\}$ and
\intervoo{m}{n} stand for $\{i \in \mathbb{N} \mid m < i < n\}$.
We use letters $\cI, \cJ$ to denote intervals and, since they stand for sets, we assume standard set relations like inclusion $\cI \subseteq \cJ$ or membership $i \in I$.

Given a set \PV\ of propositional variables (called \emph{alphabet}),
a \emph{metric formula} $\varphi$ is defined by the grammar:
\begin{eqnarray*}
\varphi &::=& \myatom \mid \bot \mid \varphi_1 \otimes \varphi_2 \mid
\metricI{\previous}\varphi \mid
\varphi_1 \metricI{\since} \varphi_2 \mid
\varphi_1 \metricI{\trigger} \varphi_2 \mid
\metricI{\next} \varphi \mid
\varphi_1 \metricI{\until} \varphi_2 \mid
\varphi_1 \metricI{\release} \varphi_2\mid \\
&&\previous \varphi \mid
\varphi_1 \since \varphi_2 \mid
\varphi_1 \trigger \varphi_2 \mid
\next \varphi \mid
\varphi_1 \until \varphi_2 \mid
\varphi_1 \release \varphi_2
\end{eqnarray*}
where $\myatom \in\PV$ is an atom, $\otimes$ is any binary Boolean connective $\otimes \in \{\to,\wedge,\vee\}$.
As usual in intuitionistic logic, we define
$\neg \varphi \eqdef \varphi \to \bot$,
$\varphi \leftrightarrow \psi \eqdef \left(\varphi \to \psi\right)\wedge \left(\psi \to \varphi\right)$ and $\top \eqdef \neg \bot$.

The binary temporal operators
$\varphi_1 \metricI{\since} \varphi_2$
and
$\varphi_1 \metricI{\trigger} \varphi_2$
denote that
$\varphi_1$ holds since or is triggered by $\varphi_2$, respectively.
Similarly,
$\varphi_1 \metricI{\until} \varphi_2$
represents the until operator, while
$\varphi_1 \metricI{\release} \varphi_2$
signifies that
$\varphi_2$ is released by $\varphi_1$.
The non-metric counterparts of these operators share the same intuitive semantics
but omit explicit reference to the interval $I$.
The unary operators
$\metricI{\next} \varphi$
and
$\metricI{\previous}\varphi$
(along with their non-metric versions $\next \varphi$ and $\previous \varphi$)
indicate that $\varphi$ holds in the next or previous state, respectively.
We distinguish between metric and non-metric variants based on the presence of interval constraints:
in the metric case,
satisfaction is contingent upon conditions being met within the intervals augmenting the operators,
whereas non-metric operators impose no such quantitative bounds.
Additionally, the following derived operators can be defined:
\[
\begin{array}{rcll}
	\metricI{\alwaysP} \varphi  & \eqdef & \bot \metricI{\trigger} \varphi                                & \text{\emph{always before}} \\
	\metricI{\eventuallyP} \varphi  & \eqdef & \top \metricI{\since} \varphi                                  & \text{\emph{eventually before}} \\
	\initially\,& \eqdef & \neg \metric{\previous}{0}{\omega} \top                        & \text{\emph{initial}}\\
	\metricI{\wprevious} \varphi  & \eqdef & \metricI{\previous} \top \to \metricI{\previous} \varphi& \text{\emph{weak previous}}
\end{array}
\quad
\begin{array}{rcll}
	\metricI{\alwaysF} \varphi  & \eqdef & \bot \metricI{\release} \varphi                                & \text{\emph{always afterward}}\\
	\metricI{\eventuallyF} \varphi  & \eqdef & \top \metricI{\until} \varphi                                  & \text{\emph{eventually afterward}}\\
	\finally  & \eqdef & \neg \metric{\next}{0}{\omega} \top                            & \text{\emph{final}}\\
	\metricI{\wnext} \varphi  & \eqdef &  \metricI{\next} \top \to \metricI{\next} \varphi        & \text{\emph{weak next}}
\end{array}
\]
\paragraph{Semantics.}
A \emph{Here-and-There trace} (for short \emph{\HT-trace}) of length $\lambda \in \mathbb{N} \cup \{\omega\}$ over alphabet \PV\ is a sequence of pairs
\(
(\tuple{H_i,T_i})_{\rangeco{i}{0}{\lambda}}
\)
with $H_i\subseteq T_i\subseteq\PV$ for any $\rangeco{i}{0}{\lambda}$.
For convenience, we often represent an \HT-trace as the pair $\tuple{\H,\T}$ of sequences
$\H = (H_i)_{\rangeco{i}{0}{\lambda}}$ and $\T = (T_i)_{\rangeco{i}{0}{\lambda}}$.
Notice that, when $\lambda=\omega$, this covers traces of infinite length.
We say that $\tuple{\H,\T}$ is \emph{total} whenever $\H=\T$, that is, $H_i=T_i$ for all $\rangeco{i}{0}{\lambda}$.

\begin{definition}[\citealp{becadiscsc24a}]\label{def:timed:trace}
	A {timed} \HT-trace $(\tuple{\H,\T},\tmf)$ of length $\lambda$ over $(\mathbb{N},<)$ and alphabet $\mathcal{A}$ is a pair consisting\nolinebreak[3] of
	\begin{itemize}
		\item an \HT-trace $\tuple{\H,\T} $ of length $\lambda$ over $\mathcal{A}$ and
		\item a function $\tmf: \intervco{0}{\lambda} \to \mathbb{N}$
		such that $\tmf(i)\leq \tmf(i{+}1)$.
	\end{itemize}
	A timed \HT-trace of length $\lambda > 1$ is called \emph{strict} if $\tmf(i)<\tmf(i{+}1)$
	for all $i{+}1 \in \intervco{0}{\lambda}$ such that $i+1 < \lambda$ and \emph{non-strict} otherwise.
	We assume w.l.o.g.\ that $\tmf(0)=0$. \qed
\end{definition}
Function \tmf\ assigns to each state index $i \in \intervco{0}{\lambda}$ a time point $\tmf(i) \in \mathbb{N}$
representing the number of time units (seconds, milliseconds, etc, depending on the chosen granularity)
elapsed since time point $\tmf(0)=0$, chosen as the beginning of the trace.
Given any timed \HT-trace,
satisfaction of formulas is defined as follows.
\begin{definition}[\MHT-satisfaction; \citealp{becadiscsc24a}]\label{def:mht:satisfaction}A timed \HT-trace $\M=(\tuple{\H,\T}, \tmf)$
	of length $\lambda$ over alphabet \PV\
	\emph{satisfies} a metric formula $\varphi$ over \PV\ at step $\rangeco{k}{0}{\lambda}$,
	written \mbox{$\M,k \models \varphi$}, if the following conditions hold:
	\begin{enumerate}
		\item $\M,k \not\models \bot$
		\item $\M,k \models \myatom$ if $\myatom \in H_k$ for any atom $\myatom \in \PV$
		\item \label{def:mhtsat:and} $\M, k \models \varphi \wedge \psi$
		iff
		$\M, k \models \varphi$
		and
		$\M, k \models \psi$
		\item \label{def:mhtsat:or} $\M, k \models \varphi \vee \psi$
		iff
		$\M, k \models \varphi$
		or
		$\M, k \models \psi$
		\item $\M, k \models \varphi \to \psi$
		iff
		$\M', k \not \models \varphi$
		or
		$\M', k \models  \psi$, for both $\M'=\M$ and $\M'=(\tuple{\T,\T}, \tmf)$
		\item \label{def:mhtsat:previous} $\M, k \models \metricI{\previous}\, \varphi$
		iff
		$k>0$ and $\M, k{-}1 \models \varphi$ and $\tmf(k)-\tmf(k{-}1) \in \cI$
		\item \label{def:mhtsat:since}$\M, k \models \varphi \metricI{\since} \psi$
		iff
		for some $\rangecc{j}{0}{k}$
		with
		$\tmf(k)-\tmf(j) \in \cI$,
		we have
		$\M, j \models \psi$
		and
		$\M, i \models \varphi$ for all $\rangeoc{i}{j}{k}$
		\item \label{def:mhtsat:trigger}$\M, k \models \varphi \metricI{\trigger} \psi$
		iff
		for all $\rangecc{j}{0}{k}$
		with
		$\tmf(k)-\tmf(j) \in \cI$,
		we have
		$\M, j \models \psi$
		or
		$\M, i \models \varphi$ for some $\rangeoc{i}{j}{k}$
		\item \label{def:mhtsat:next}$\M, k \models \metricI{\next}\, \varphi$
		iff
		$k+1<\lambda$ and $\M, k{+}1 \models \varphi$
		and $\tmf(k{+}1)-\tmf(k) \in \cI$
		\item \label{def:mhtsat:until}$\M, k \models \varphi \metricI{\until} \psi$
		iff
		for some $\rangeco{j}{k}{\lambda}$
		with
		$\tmf(j)-\tmf(k) \in \cI$,
		we have
		$\M, j \models \psi$
		and
		$\M, i \models \varphi$ for all $\rangeco{i}{k}{j}$
		\item \label{def:mhtsat:release} $\M, k \models \varphi \metricI{\release} \psi$
		iff
		for all $\rangeco{j}{k}{\lambda}$
		with
		$\tmf(j)-\tmf(k) \in \cI$,
		we have
		$\M, j \models \psi$
		or
		$\M, i \models \varphi$ for some $\rangeco{i}{k}{j}$\qed
\end{enumerate}
\end{definition}

A formula $\varphi$ is a \emph{tautology} (or is valid), written $\models \varphi$, iff
$\M,k \models \varphi$ for any timed \HT-trace $\M$ and any $k \in \intervco{0}{\lambda}$.
\MHT\ is the logic induced by the set of all such tautologies.
For two formulas $\varphi, \psi$ we write $\varphi \equiv \psi$, iff
$\models \varphi \leftrightarrow \psi$, that is,
$\M,k \models \varphi \leftrightarrow \psi$ for any
timed \HT-trace \M\ of length $\lambda$ and any $k \in \intervco{0}{\lambda}$.
A timed \HT-trace $\M$ is an \MHT\ \emph{model} of a metric temporal formula $\varphi$ if $\M,0 \models \varphi$.
Similarly, $\M$ is an \MHT\ \emph{model} of theory $\Gamma$ if $\M,0 \models \varphi$ for all $\varphi \in \Gamma$.
The set of \MHT\ models of $\Gamma$ having length $\lambda$ is denoted as $\MHT(\Gamma,\lambda)$,
whereas $\MHT(\Gamma) \eqdef \bigcup_{\lambda=0}^\omega \MHT(\Gamma,\lambda)$ is
the set of all \MHT\ models of $\Gamma$ of any length.
We may obtain fragments of any metric logic by imposing restrictions on the timed \HT-traces used for defining tautologies and models.
That is, \MHTf\ stands for the restriction of \MHT\ to traces of any finite length $\lambda \in \mathbb{N}$ and
\MHTo\ corresponds to the restriction to traces of infinite length $\lambda=\omega$.

An interesting subset of \MHT\ is the one formed by \empty{total} timed \HT-traces like $(\tuple{\T, \T}, \tmf)$.
In the non-metric version of temporal \HT, the restriction to
total models corresponds to Linear Temporal Logic (\LTL,~\cite{pnueli77a}).
In our case, the restriction to total traces defines a metric version of \LTL,
which we call \emph{Metric Temporal Logic} (or \MTL\ for short).
\cite{becadiscsc24a} show that
\MHT\ satisfies the common persistence and negation properties, and that
\MHTf\ is decidable.

\paragraph{Strict Traces.}
Next, we consider a group of results that hold under the assumption of strict traces, namely,
that $\tmf(i) < \tmf(i+1)$ for any pair of consecutive time points.
We can enforce metric models to be traces with a strict timing function $\tmf$.
This can be achieved with the addition of the formula $\alwaysF \neg \metric{\next}{0}{0} \top$ to the context,
where $0$ abbreviates \intervcc{0}{0}.
\begin{proposition}\label{prop:mht:strict}
  For any \MHT\ trace $(\tuple{\Htrace,\Ttrace},\tau)$ of length $\lambda\in \mathbb{N}\cup \lbrace \omega \rbrace$,
  we have that $(\tuple{\Htrace,\Ttrace},\tau) \models \alwaysF\neg \metric{\next}{0}{0} \top$
  iff $\tau$ is a strict function.\qed
\end{proposition}
In Lemma~\ref{lemma_inductive_def_interval_zero} below,
the use of strict traces enables a recursive definition for
operators of the form $\metric{Q}{0}{n}$, with $Q \in \lbrace \until, \release,\since,\trigger\rbrace$.
This novel characterization, similar to the recursive definition of binary temporal operators in \LTL~\citep{pnueli77a},
is more suitable for the Tseitin-style translation developed in this paper.
\label{rev3.4}For this reason, we henceforth restrict our focus to strict traces,
allowing us to assume the implicit inclusion of the aforementioned axiom.
The following equivalences state that the interval \intervcc{0}{0} makes all binary metric operators
collapse into their right hand argument formula,
whereas unary operators collapse to a truth constant.
For metric formulas $\psi$ and $\varphi$ and for strict traces, we have:
\begin{align}
	\varphi\metric{\until}{0}{0}\psi &\equiv \varphi\metric{\release}{0}{0}\psi  \equiv \varphi\metric{\since}{0}{0}\psi \equiv \varphi\metric{\trigger}{0}{0}\psi \equiv \psi \label{inductive_def_zero:1}
	\\
	\metric{\next}{0}{0} \,  \varphi &\equiv \metric{\previous}{0}{0} \,  \varphi \equiv \bot    \label{inductive_def_zero:2}
	\\
	\metric{\wnext}{0}{0} \, \varphi &\equiv \metric{\wprevious}{0}{0} \, \varphi \equiv \top    \label{inductive_def_zero:3}
\end{align}
The last two lines are precisely an effect of dealing with strict traces.
For instance, $\metric{\next}{0}{0} \, \varphi \equiv \bot$ tells us that
it is always impossible to have a successor state with the same time as the current one,
regardless of the formula $\varphi$ at hand.
Lemma~\ref{lemma_inductive_def_interval_zero} below allows unfolding \emph{until}, \emph{release}, \emph{since} and \emph{trigger}
for intervals of the form $\intervcc{0}{n}$.
\begin{lemma}[\citealp{becadiscsc24a}]\label{lemma_inductive_def_interval_zero}
	For metric formulas $\psi$ and $\varphi$ and for $n \ge 0$, we have:
	\begin{align}
		\label{def:inductive_until_0n} \varphi \metric{\until}{0}{n} \psi & {} \equiv\textstyle \psi \vee ( \varphi \wedge \bigvee_{x=1}^{n} \metric{\next}{x}{x} ( \varphi \metric{\until}{0}{(n-x)} \psi ) ) \\
		\label{def:inductive_release_0n} \varphi \metric{\release}{0}{n} \psi & {} \equiv\textstyle \psi \wedge ( \varphi \vee \bigwedge_{x=1}^{n} \metric{\wnext}{x}{x} ( \varphi \metric{\release}{0}{(n-x)} \psi ) )\\
		\label{def:inductive_since_0n} \varphi \metric{\since}{0}{n} \psi & {} \equiv\textstyle \psi \vee ( \varphi \wedge \bigvee_{x=1}^{n} \metric{\previous}{x}{x} ( \varphi \metric{\since}{0}{(n-x)} \psi ) )\\
		\label{def:inductive_trigger_0n} \varphi \metric{\trigger}{0}{n} \psi & {} \equiv\textstyle \psi \wedge ( \varphi \vee \bigwedge_{x=1}^{n} \metric{\wprevious}{x}{x} ( \varphi \metric{\trigger}{0}{(n-x)} \psi ) ),
	\end{align}
        where intervals of the form $x$ stand for $\intervcc{x}{x}$.\qed
\end{lemma}
The equivalences presented above hold under the assumption of strict traces.
To show it, let us consider the formula $\press  \metric{\until}{0}{1}\arrives$ and
let us define the timed \HT-trace $\M=(\tuple{\Ttrace,\Ttrace},\tau)$ of length $\lambda=2$ as
$T_0 = \lbrace \press \rbrace$, $T_1 = \lbrace \arrives \rbrace$ and $\tau(0)=\tau(1)=0$.
Since $\tau(0) \not < \tau(1)$, $\M$ is not a strict trace.
Hence, $\M, 0 \models \press  \metric{\until}{0}{1}\arrives$ but
$\M, 0 \not \models \arrives \vee \left(\press \wedge \metric{\next}{1}{1}\left(\press \metric{\until}{0}{0} \arrives \right)\right)$,
which falsifies~\eqref{def:inductive_until_0n}.
Formulas of the form $\varphi \metric{\until}{a}{b}\psi$,
where $[a..b)$ is a general interval,
accept more complicated recursive definitions (see~\citep{becadiscsc24a} for details) and they are left out of the scope of this paper.

\paragraph{Metric Equilibrium Models.} As in (propositional) Equilibrium Logic~\citep{pearce06a},
non-monotonicity is achieved by means of a selection among the timed \HT-traces of a theory.
\begin{definition}[Metric Equilibrium Model~\citep{becadiscsc24a}]
  Let $\mathfrak{G}$ be some set of timed \HT-traces of a given theory.
  A total timed \HT-trace $(\tuple{\Ttrace,\Ttrace},\tau)$ is a metric equilibrium model of $\mathfrak{G}$
  iff
  there is no other $\Htrace$ different from $\Ttrace$ such that $(\tuple{\Htrace,\Ttrace},\tau) \in \mathfrak{G}$.
\end{definition}
We talk about metric equilibrium models of a theory $\Gamma$ when $\mathfrak{S} = \MHT(\Gamma)$, and
we write $\MEL(\Gamma, \lambda)$ and $\MEL(\Gamma)$ to stand for
the metric equilibrium models of $\MHT(\Gamma, \lambda)$ and $\MHT(\Gamma)$, respectively.
\emph{Metric Equilibrium Logic} (MEL) is the non-monotonic logic induced by the metric equilibrium models of metric temporal theories.
Variants \MELf\ and \MELo\ refer to \MEL\ when restricted to traces of finite and infinite length, respectively.
 \section{Three-valued Semantics}\label{sec:three-valued}
It is well known that \HT{} semantics is equivalent to three-valued G\"odel logic~\citep{goedel32a}.
The use of the latter semantics in the context of Equilibrium Logic and its temporal extensions facilitates correctness proofs of Tseitin reductions to logic programs and temporal logic programs, 
since it allows equivalences to be evaluated more conveniently.
Three-valued G\"odel logic was extended to the temporal setting in~\cite{cabalar10a}, within the context of reducing temporal theories to temporal logic programs.
Since, in this paper, we extend Cabalar's results to the metric case, we also provide a three-valued characterisation of \MHT{}.
In our setting, a three-valued interpretation is a mapping  
\begin{equation*}
m: \left(\intervco{0}{\lambda}\to \mathbb{N}\right) \to  \left(\left(\intervco{0}{\lambda} \times \PV\right) \to \lbrace 0,1,2\rbrace\right)
\end{equation*}
where, broadly speaking, the first input function $\left( \intervco{0}{\lambda}\to \mathbb{N}\right)$
represents a timing function that
gives the different time points associated with each state of the trace;
we denote it by $\tau$ in the rest of this section.
Given $\tau$, we define 
\begin{equation*}
	m(\tau): \intervco{0}{\lambda} \times \PV \to \lbrace 0,1,2\rbrace
\end{equation*}
as the associated evaluation for propositional variables, that is,
a function that provides the truth value for each propositional variable $p \in \PV$ at each time point $\rangeco{k}{0}{\lambda}$. Such $m(\tau)$\footnote{Note that here $\tau$ is implicit.} can be extended to $\bm{m}$ to cope with arbitrary formulas as shown below.
\begingroup
\allowdisplaybreaks
\begin{align*}
	\trival{k}{\bot} & =  0 \\
	\trival{k}{a} &= m(\tau)(k,a) \hspace{20pt} \text{for any atom } a\\
	\trival{k}{\varphi \wedge \psi} & = 
	\min(\trival{k}{\varphi},\trival{k}{\psi})\\
	\trival{k}{\varphi \vee   \psi} & = 
	\max(\trival{k}{\varphi},\trival{k}{\psi})\\
	\trival{k}{\varphi \to \psi} & =  \begin{cases}
											2 & \text{ if } \trival{k}{\varphi} \le \trival{k}{\psi}\\
											\trival{k}{\psi} & \text{ otherwise }
											\end{cases}\\
	\trival{k}{\metricI{\previous} \varphi} & = 
	\begin{cases}
		0                      & \text{if } k=0 \text{ or } \tau(k)-\tau(k-1) \not \in I\\
		\trival{k-1}{\varphi}  & \text{otherwise }
	\end{cases}\\
	\trival{k}{\metricI{\wprevious} \varphi} & = 
	\begin{cases}
		2                      & \text{if } k=0 \text{ or } \tau(k)-\tau(k-1) \not \in I\\
		\trival{k-1}{\varphi}  & \text{otherwise }
	\end{cases}\\
\trival{k}{\varphi \metricI{\since} \psi} & = 
\sup\{\min\{\trival{i}{\psi},\trival{j}{\varphi}\mid \rangeoc{j}{i}{k}\} \mid \rangecc{i}{0}{k} \text{ and } \tau(k)-\tau(i) \in I\}\footnotemark\\
\trival{k}{\varphi \metricI{\trigger} \psi} & = 
	\inf\{\max\{\trival{i}{\psi},\trival{j}{\varphi}\mid \rangeoc{j}{i}{k}\} \mid \rangecc{i}{0}{k} \text{ and } \tau(k)-\tau(i)\in I\}\\
	\trival{k}{\metricI{\next} \varphi} & = 
	\begin{cases}
		0                      & \text{if } k+1=\lambda \ (\neq \omega) \text{ or } \tau(k+1)-\tau(k)\not \in I\\
		\trival{k+1}{\varphi}  & \text{otherwise }
	\end{cases}\\
	\trival{k}{\metricI{\wnext} \varphi} & = 
	\begin{cases}
		2                    & \text{if } k+1=\lambda \ (\neq \omega) \text{ or } \tau(k+1)-\tau(k)\not \in I\\
		\trival{k+1}{\varphi}  & \text{otherwise }
	\end{cases}\\
	\trival{k}{\varphi \metricI{\until} \psi} & =  \sup\{\min\{\trival{i}{\psi},\trival{j}{\varphi}\mid \rangeco{j}{k}{i}\}\mid \rangeco{i}{k}{\lambda} \text{ and } \tau(i)-\tau(k)\in I\}\\
	\trival{k}{\varphi \metricI{\release} \psi} & =  \inf\{\max\{\trival{i}{\psi},\trival{j}{\varphi}\mid \rangeco{j}{k}{i}\}\mid\rangeco{i}{k}{\lambda} \text{ and } \tau(i)-\tau(k)\in I\}\\
\end{align*}
\endgroup
\footnotetext{Here we consider that $\sup \emptyset = 0$ and $\inf \emptyset = 2$.}The three-valued semantics defined above offers a significant advantage when establishing equivalences.
For instance,
demonstrating that $\trival{k}{\varphi \leftrightarrow \psi}=2$ is equivalent
to showing that $\trival{k}{\varphi} = \trival{k}{\psi}$
for all formulas $\varphi$ and $\psi$.
Similarly,
proving that $\trival{0}{\alwaysF(\varphi\leftrightarrow \psi)}$ is equivalent
to establishing that $\trival{k}{\varphi}=\trival{k}{\psi}$ for all $\rangeco{k}{0}{\lambda}$.
These properties streamline the formal verification process,
as they allow complex logical equivalences to be treated as identities within our three-valued framework.

To establish the equivalence between the standard \MHT\ semantics and our three-valued characterization,
we define a function $f$
from the set of timed \HT{}-traces to the set of three-valued interpretations.
For any timed \HT-trace $\M = (\tuple{\H,\T},\tau)$ of length $\lambda$,
let
\(
f(\M) = m(\tau)
\)
where
\(
m(\tau): [0 \cdot \cdot \lambda) \times \PV \to \lbrace 0,1,2\rbrace
\)
is defined for any $\rangeco{k}{0}{\lambda}$ and $a \in \PV$ as follows:
\begin{equation*}
	m(\tau)(k,a) = \begin{cases}
		0 & \text{if} \ a \not\in T_k \\
		1 & \text{if} \ a \in T_k\setminus H_k \\
		2 & \text{if} \ a \in H_k
	\end{cases}\\
\end{equation*}
This function provides us with a 1-1 correspondence between timed \HT{}-traces and three-valued interpretations.
It can be checked that $f$ is \emph{bijective}.
The following lemma shows that the satisfaction of arbitrary metric temporal formulas is preserved among the semantics.
\begin{proposition}\label{prop:equivalence:semantics}
	For every timed \HT{}-trace $\M$ of length $\lambda$, every $\rangeco{k}{0}{\lambda}$ and every formula $\varphi$, 
	$\M, k\models \varphi$ iff $f(\M)(k,\varphi) = 2$ and $(\tuple{\Ttrace,\Ttrace}), k\models \varphi$ iff $f(\M)(k,\varphi) \not = 0$.\qed
\end{proposition}
\begin{corollary}
The \MHT\ semantics and its three-valued characterization are equivalent.\qed
\end{corollary}

For our Tseitin translation,
we impose the following restriction:
the intervals associated to binary temporal operators $\metricI{\since}$,
$\metricI{\trigger}$, $\metricI{\until}$ and $\metricI{\release}$ are of the form $I=\intervcc{0}{n}$,
with $n<\omega$.
By abbreviating the interval $\intervcc{0}{n}$ by `${\le}\,n$',
the metric temporal operators are of the form
$\metric{\since}{0}{n}$,$\metric{\trigger}{0}{n}$,$\metric{\until}{0}{n}$,$\metric{\release}{0}{n}$.
The following proposition shows how equivalences~\eqref{def:inductive_until_0n}-\eqref{def:inductive_trigger_0n}
lead to an alternative three-valued interpretation for the metric operators.
\begin{proposition}\label{prop:alternative:semantics}
For any interpretation $\bm{m}$ of length $\lambda$, any $\rangeco{k}{0}{\lambda}$ and any $\rangeco{n}{0}{\omega}$,  we have:
\begin{align*}
\trival{k}{\varphi \metric{\until}{0}{n}\psi }   &= \max\lbrace \trival{k}{\psi},\min\lbrace \trival{k}{\varphi}, \sup\lbrace \trival{k}{\metric{\next}{x}{x} ( \varphi \metric{\until}{0}{(n-x)} \psi )} \mid \rangecc{x}{1}{n} \rbrace \rbrace \rbrace \\
\trival{k}{\varphi \metric{\release}{0}{n} \psi} &= \min\lbrace \trival{k}{\psi},\max\lbrace \trival{k}{\varphi}, \inf\lbrace \trival{k}{\metric{\wnext}{x}{x} ( \varphi \metric{\release}{0}{(n-x)} \psi )} \mid \rangecc{x}{1}{n} \rbrace \rbrace \rbrace \\
\trival{k}{\varphi \metric{\since}{0}{n} \psi} &= \max\lbrace \trival{k}{\psi},\min\lbrace \trival{k}{\varphi}, \sup\lbrace \trival{k}{\metric{\previous}{x}{x} ( \varphi \metric{\since}{0}{(n-x)} \psi )} \mid \rangecc{x}{1}{n} \rbrace \rbrace \rbrace  \\
\trival{k}{\varphi \metric{\trigger}{0}{n}\psi}  &= \min\lbrace \trival{k}{\psi},\max\lbrace \trival{k}{\varphi}, \inf\lbrace \trival{k}{\metric{\wprevious}{x}{x} ( \varphi \metric{\trigger}{0}{(n-x)} \psi )} \mid \rangecc{x}{1}{n} \rbrace \rbrace \rbrace
\end{align*}\qed
\end{proposition}
Note that, by applying Proposition~\ref{prop:alternative:semantics}, we can conclude that $\trival{k}{\varphi \metric{\until}{0}{0}\psi} = \trival{k}{\varphi \metric{\release}{0}{0}\psi} = \trival{k}{\varphi \metric{\since}{0}{0}\psi} = \trival{k}{\varphi \metric{\trigger}{0}{0}\psi} = \trival{k}{\psi}$.
In what follows,
we use those semantics,
when dealing with metric modalities of the form $\metric{Q}{0}{n}$
with $Q \in \lbrace\until, \release,\since,\trigger \rbrace$.
 
\section{Translation into Metric Temporal Logic Programs}\label{sec:translation}

\begin{definition}[Metric temporal literal, rule and program]
Given an alphabet $\PV$, we define the set of \emph{regular literals} as $\lbrace a, \neg a\mid a \in \PV\rbrace$, the set of \emph{metric temporal literals} as
$\lbrace\metricI{\previous} a,\neg \metricI{\previous} a, \metricI{\wprevious} a, \neg \metricI{\wprevious} a, \previous a,\neg \previous a, \wprevious a ,\neg \wprevious a  \mid a \in \PV\rbrace$.
A \emph{metric temporal rule} is either
\begin{itemize}
	\item an \emph{initial rule} $B\to A$
	\item a \emph{dynamic rule} of the form $\wnext\alwaysF \left( B\to A\right) $, or
	\item a \emph{final rule} of the form $\alwaysF \left(\finally \to \left(B \to A\right)\right)$
\end{itemize}
where $B= \bigwedge\limits_{i=1}^n b_i$ with $n\ge 0$,
$A=\bigvee\limits_{j=1}^m a_j$ with $m\ge 0$, and
each $b_i$ and $a_j$ are metric temporal literals for dynamic rules, or
regular literals for initial and final rules.
Facts of the form $p$ are equivalent to rules of the form $\to p$ while
constraints of the form $\neg p$ are equivalent to rules of the type $p \to \bot$.
A \emph{metric temporal logic program} is a set of metric temporal rules.\qed
\end{definition}
\begin{table}[h!]\centering
\caption{Translation for metric next, weak next, previous and weak previous.}
\label{tbl:translation:mnext-mprevious}
{\tablefont \begin{tabular}{@{\extracolsep{\fill}}ccc}
\topline
$\mu$ & $\eta(\mu)$ & $\eta^*(\mu)$ \\\hline
$\metricI{\next} \varphi$ & $ \begin{array}{l}
	\wnext\alwaysF\left( \metricI{\previous} \top \to \left(\previous \Lab{\mu} \leftrightarrow \Lab{\varphi}\right) \right)\\
	\wnext\alwaysF\left( \neg \metricI{\previous} \top \to \neg \previous \Lab{\mu} \right)\\
	\alwaysF\left( \finally \to \neg \Lab{\mu}\right)
\end{array}$ & $\begin{array}{l}
\wnext\alwaysF\left( \metricI{\previous} \Lab{\mu} \rightarrow \Lab{\varphi} \right)\\
\wnext\alwaysF\left( \metricI{\previous} \top \wedge \Lab{\varphi} \rightarrow \previous \Lab{\mu} \right)\\
	\wnext\alwaysF\left( \neg \metricI{\previous} \top \to \neg \previous \Lab{\mu} \right)\\
\alwaysF\left( \finally \to \neg \Lab{\mu}\right)
\end{array}$ \\\hline
$\metricI{\wnext} \varphi$ & $\begin{array}{l}
	\wnext\alwaysF\left( \metricI{\previous} \top \to \left( \previous \Lab{\mu} \leftrightarrow \Lab{\varphi}\right) \right)\\
	\wnext\alwaysF\left( \neg \metricI{\previous} \top \to  \previous \Lab{\mu} \right)\\
	\alwaysF\left( \finally \to  \Lab{\mu}\right)
\end{array}$ & $\begin{array}{l}
	\wnext\alwaysF\left( \metricI{\previous} \Lab{\mu} \rightarrow \Lab{\varphi} \right)\\
	\wnext\alwaysF\left( \metricI{\previous} \top \wedge \Lab{\varphi} \rightarrow  \previous \Lab{\mu} \right)\\
	\wnext\alwaysF\left( \neg \metricI{\previous} \top \to  \previous \Lab{\mu} \right)\\
	\alwaysF\left( \finally \to  \Lab{\mu}\right)
\end{array}$
\\\hline
$\metricI{\previous} \varphi $ & $\begin{array}{l}
	\wnext\alwaysF\left(\Lab{\mu} \leftrightarrow  \metricI{\previous}\Lab{\varphi}\right)\\
	\neg \Lab{\mu}
\end{array}$ & $\begin{array}{l}
\wnext\alwaysF\left( \Lab{\mu}  \rightarrow \metricI{\previous}\Lab{\varphi}\right)\\
\wnext\alwaysF\left( \metricI{\previous}\Lab{\varphi} \to \Lab{\mu} \right)\\
\neg \Lab{\mu}
\end{array}$\\\hline
$\metricI{\wprevious} \varphi $ & $\begin{array}{l}
	\wnext\alwaysF\left(\Lab{\mu} \leftrightarrow  \metricI{\wprevious}\Lab{\varphi}\right)\\
    \Lab{\mu}
\end{array}$ & $\begin{array}{l}
\wnext\alwaysF\left( \Lab{\mu} \wedge \metricI{\wprevious}\top  \rightarrow \previous \Lab{\varphi}\right)\\
\wnext\alwaysF\left( \metricI{\wprevious}\Lab{\varphi} \to \Lab{\mu} \right)\\
\Lab{\mu}
\end{array}$\botline
\end{tabular}
}
\end{table}
 \begin{table}[h!]
\caption{Translation of metric since and metric trigger.} \label{tbl:translation:msince-mtrigger}
\resizebox{\textwidth}{!}{
{\tablefont
\begin{tabular}{@{\extracolsep{\fill}}ccc}\topline
$\mu$ & $\eta(\mu)$ & $\eta^*(\mu)$ \\\hline
$\varphi \metric{\since}{0}{n} \psi$ & $\begin{array}{l}
\alwaysF\left(\Lab{\mu}\leftrightarrow \left(\Lab{\psi} \vee   \left(\Lab{\varphi} \wedge \bigvee\limits_{x=1}^n \Lab{\metric{\previous}{x}{x} (\varphi\metric{\since}{0}{(n-x)} \psi)}\right)\right)\right) \\
\end{array}$ & $\begin{array}{l}
\Lab{\mu}\rightarrow \Lab{\psi} \vee \Lab{\varphi}\\
\Lab{\mu}\rightarrow \Lab{\psi} \vee \bigvee\limits_{x=1}^n \Lab{\metric{\previous}{x}{x} (\varphi\metric{\since}{0}{(n-x)} \psi)}\\
\Lab{\psi}\rightarrow \Lab{\mu} \\
\Lab{\varphi} \wedge \Lab{\metric{\previous}{1}{1} (\varphi\metric{\since}{0}{(n-1)} \psi)} \rightarrow \Lab{\mu} \\
\cdots \\
\Lab{\varphi} \wedge \Lab{\metric{\previous}{n}{n} (\varphi\metric{\since}{0}{(n-n)} \psi)} \rightarrow \Lab{\mu} \\
\wnext\alwaysF\left(\Lab{\mu}\rightarrow \Lab{\psi} \vee \Lab{\varphi} \right) \\
\wnext\alwaysF\left(\Lab{\mu}\rightarrow \Lab{\psi} \vee \bigvee\limits_{x=1}^n \Lab{\metric{\previous}{x}{x} (\varphi\metric{\since}{0}{(n-x)} \psi)}\right) \\
\wnext\alwaysF\left(\Lab{\psi}\rightarrow \Lab{\mu} \right) \\
\wnext\alwaysF\left(\Lab{\varphi} \wedge \Lab{\metric{\previous}{1}{1} (\varphi\metric{\since}{0}{(n-1)} \psi)} \rightarrow \Lab{\mu}\right) \\
\cdots \\
\wnext\alwaysF\left(\Lab{\varphi} \wedge \Lab{\metric{\previous}{n}{n} (\varphi\metric{\since}{0}{(n-n)} \psi)} \rightarrow \Lab{\mu}\right) \\
\end{array}$\\\hline
$\varphi \metric{\trigger}{0}{n} \psi$ & $\begin{array}{l}
\alwaysF\left(\Lab{\mu}\leftrightarrow \left(\Lab{\psi}\wedge \left( \Lab{\varphi} \vee \bigwedge\limits_{x=1}^n \Lab{\metric{\wprevious}{x}{x} (\varphi\metric{\trigger}{0}{(n-x)} \psi)}\right)\right)\right) \\
\end{array}$ & $\begin{array}{l}
\Lab{\mu}\rightarrow \Lab{\psi} \\
\Lab{\mu}\rightarrow \Lab{\varphi} \vee  \Lab{\metric{\wprevious}{1}{1} (\varphi\metric{\trigger}{0}{(n-1)} \psi)} \\
\cdots \\
\Lab{\mu}\rightarrow \Lab{\varphi} \vee  \Lab{\metric{\wprevious}{n}{n} (\varphi\metric{\trigger}{0}{(n-n)} \psi)} \\
\Lab{\psi} \wedge \Lab{\varphi} \to \Lab{\mu} \\
\left(\Lab{\psi}\wedge  \bigwedge\limits_{x=1}^n \Lab{\metric{\wprevious}{x}{x} (\varphi\metric{\trigger}{0}{(n-x)} \psi)}\right) \to \Lab{\mu} \\
\wnext\alwaysF\left(\Lab{\mu}\rightarrow \Lab{\psi}\right) \\
\wnext\alwaysF\left(\Lab{\mu}\rightarrow \Lab{\varphi} \vee  \Lab{\metric{\wprevious}{1}{1} (\varphi\metric{\trigger}{0}{(n-1)} \psi)}\right) \\
\cdots \\
\wnext\alwaysF\left(\Lab{\mu}\rightarrow \Lab{\varphi} \vee  \Lab{\metric{\wprevious}{n}{n} (\varphi\metric{\trigger}{0}{(n-n)} \psi)}\right) \\
\wnext\alwaysF\left( \Lab{\psi} \wedge \Lab{\varphi} \to \Lab{\mu}\right) \\
\wnext\alwaysF\left(\left(\Lab{\psi}\wedge  \bigwedge\limits_{x=1}^n \Lab{\metric{\wprevious}{x}{x} (\varphi\metric{\trigger}{0}{(n-x)} \psi)}\right) \to \Lab{\mu}\right) \\
\end{array}$\botline
\end{tabular}}}
\end{table}
 \begin{table}[h!]
\caption{Translation of metric until and metric release.}
\label{tbl:translation:muntil-mrelease}
\resizebox{.9\textwidth}{!}{
{\tablefont \begin{tabular}{@{\extracolsep{\fill}}ccc}\topline
$\mu$ & $\eta(\mu)$ & $\eta^*(\mu)$ \\\hline	
$\varphi \metric{\until}{0}{n} \psi$ & $\begin{array}{l} 
	\alwaysF\left(\Lab{\mu}\leftrightarrow \Lab{\psi} \vee \left( \Lab{\varphi} \wedge \bigvee\limits_{x=1}^n \Lab{\metric{\next}{x}{x} (\varphi\metric{\until}{0}{(n-x)} \psi)}\right)\right) \\
\end{array}$ &$\begin{array}{l} 
\Lab{\mu}\rightarrow \Lab{\psi} \vee \Lab{\varphi} \\
\Lab{\mu}\rightarrow \Lab{\psi} \vee \bigvee\limits_{x=1}^n \Lab{\metric{\next}{x}{x} (\varphi\metric{\until}{0}{(n-x)} \psi)} \\
\Lab{\psi}\rightarrow \Lab{\mu} \\
\Lab{\varphi} \wedge \Lab{\metric{\next}{1}{1} (\varphi\metric{\until}{0}{(n-1)} \psi)} \rightarrow \Lab{\mu}\\
\cdots\\
\Lab{\varphi} \wedge \Lab{\metric{\next}{n}{n} (\varphi\metric{\until}{0}{(n-n)} \psi)} \rightarrow \Lab{\mu}\\
\wnext\alwaysF\left(\Lab{\mu}\rightarrow \Lab{\psi} \vee \Lab{\varphi} \right)\\
\wnext\alwaysF\left(\Lab{\mu}\rightarrow \Lab{\psi} \vee \bigvee\limits_{x=1}^n \Lab{\metric{\next}{x}{x} (\varphi\metric{\until}{0}{(n-x)} \psi)}\right) \\
\wnext\alwaysF\left(\Lab{\psi}\rightarrow \Lab{\mu}\right) \\
\wnext\alwaysF\left(\Lab{\varphi} \wedge \Lab{\metric{\next}{1}{1} (\varphi\metric{\until}{0}{(n-1)} \psi)} \rightarrow \Lab{\mu}\right) \\
\cdots\\
\wnext\alwaysF\left(\Lab{\varphi} \wedge \Lab{\metric{\next}{n}{n} (\varphi\metric{\until}{0}{(n-n)} \psi)} \rightarrow \Lab{\mu}\right) \\
\end{array}$ \\\hline
$\varphi \metric{\release}{0}{n} \psi$ & $\begin{array}{l} 
\alwaysF\left(\Lab{\mu}\leftrightarrow \Lab{\psi} \wedge \left( \Lab{\varphi} \vee \bigwedge\limits_{x=1}^n \Lab{\metric{\wnext}{x}{x} (\varphi\metric{\release}{0}{(n-x)} \psi)}\right)\right) \\
\end{array}$ & $\begin{array}{l} 
\Lab{\mu}\rightarrow \Lab{\psi} \\
\Lab{\mu}\rightarrow   \Lab{\varphi} \vee \Lab{\metric{\wnext}{1}{1} (\varphi\metric{\release}{0}{(n-1)} \psi)} \\
\cdots \\
\Lab{\mu}\rightarrow   \Lab{\varphi} \vee \Lab{\metric{\wnext}{n}{n} (\varphi\metric{\release}{0}{(n-n)} \psi)}\\
\Lab{\psi}\wedge \Lab{\varphi} \rightarrow \Lab{\mu}\\
\left( \Lab{\psi} \wedge \bigwedge\limits_{x=1}^n \Lab{\metric{\wnext}{x}{x} (\varphi\metric{\release}{0}{(n-x)} \psi)}\right)\rightarrow \Lab{\mu} \\
\wnext\alwaysF\left(\Lab{\mu}\rightarrow \Lab{\psi} \right) \\
\wnext\alwaysF\left(\Lab{\mu}\rightarrow   \Lab{\varphi} \vee \Lab{\metric{\wnext}{1}{1} (\varphi\metric{\release}{0}{(n-1)} \psi)}\right) \\
\cdots \\
\wnext\alwaysF\left(\Lab{\mu}\rightarrow   \Lab{\varphi} \vee \Lab{\metric{\wnext}{n}{n} (\varphi\metric{\release}{0}{(n-n)} \psi)}\right) \\
\wnext\alwaysF\left(\Lab{\psi}\wedge \Lab{\varphi} \rightarrow \Lab{\mu}\right) \\
\wnext\alwaysF\left(\left( \Lab{\psi} \wedge \bigwedge\limits_{x=1}^n \Lab{\metric{\wnext}{x}{x} (\varphi\metric{\release}{0}{(n-x)} \psi)}\right)\rightarrow \Lab{\mu} \right) \\
\end{array}$\botline
\end{tabular}}}
\end{table}

\label{r3.12}In the propositional setting,
the Tseitin transformation is typically performed by substituting subformulas with auxiliary atoms that
represent their respective truth values.
In the metric temporal case,
the recursive nature of the metric temporal operators necessitates a more complex approach.
For instance, if we follow Proposition~\ref{lemma_inductive_def_interval_zero}, the satisfaction of the formula $\moving \metric{\until}{0}{3} \arrives$ requires the evaluation of $\metric{\next}{x}{x}\left(\moving\metric{\until}{0}{n-x}\arrives\right)$, for all $\rangecc{x}{1}{3}$, which are not subformulas of $\moving \metric{\until}{0}{3} \arrives$.

Inspired by the Fisher-Ladner closure used in dynamic logics~\citep{fislad79a}, the set of subformulas must be closed under the recursive definitions of the binary metric operators.
This closure ensures that the auxiliary atoms correctly capture the temporal expansion of formulas
across the discrete time points of the timed trace.
\begin{definition}[Closure] Let $\varphi$ be a metric temporal formula.
  We define the \emph{closure} of $\varphi$, in symbols $cl(\varphi)$,
  as the smallest theory
  satisfying the following properties:
	\begin{enumerate}
		\item $\varphi \in cl(\varphi)$
		\item $cl(\varphi)$ is closed under subformulas
\item $\varphi\metric{\until}{0}{n} \psi \in cl(\varphi)$ and $n\ge 1$ imply $\next_i\left(\varphi\metric{\until}{0}{n-i} \psi\right) \in cl(\varphi)$, for all $\rangecc{i}{1}{n}$
		\item $\varphi\metric{\release}{0}{n} \psi \in cl(\varphi)$ and $n\ge 1$ imply $\wnext_i\left(\varphi\metric{\release}{0}{n-i} \psi\right) \in cl(\varphi)$, for all $\rangecc{i}{1}{n}$
		\item $\varphi\metric{\since}{0}{n} \psi \in cl(\varphi)$ and  $n\ge 1$ imply $\previous_i\left(\varphi\metric{\since}{0}{n-i} \psi\right) \in cl(\varphi)$, for all $\rangecc{i}{1}{n}$
		\item $\varphi\metric{\trigger}{0}{n} \psi \in cl(\varphi)$ and  $n\ge 1$ imply $\wprevious_i \left(\varphi\metric{\trigger}{0}{n-i} \psi\right) \in cl(\varphi)$, for all $\rangecc{i}{1}{n}$\qed
	\end{enumerate}
\end{definition}
For instance, the closure of the formula~\eqref{ex:press:moving}, $\alwaysF (\mathit{press} \to \mathit{moving} \metric{\until}{0}{3}\mathit{arrives})$,
corresponds to the set $\mathit{cl}(\eqref{ex:press:moving})$ below.
For brevity, we abbreviate $\mathit{press}$, $\mathit{moving}$ and $\mathit{arrives}$
as $\press$, $\moving$ and $\arrives$, respectively.
\begin{equation*}
\mathit{cl}(\eqref{ex:press:moving}) = \left\{
										\begin{array}{rrr}
										\alwaysF (\press \to \moving \metric{\until}{0}{3}\arrives)\mytag{(F1)}, & \myex \mytag{(F2)}, &
										\moving\metric{\until}{0}{3}\arrives \mytag{(F3)},\\
										\metric{\next}{3}{3}\!\left(\moving \metric{\until}{0}{0}\arrives\right) \mytag{(F4)},&
										\metric{\next}{2}{2}\!\left(\moving \metric{\until}{0}{1}\arrives\right) \mytag{(F5)}, & \metric{\next}{2}{2}\!\left(\moving \metric{\until}{0}{0}\arrives\right) \mytag{(F6)}, \\
										\metric{\next}{1}{1}\!\left(\moving \metric{\until}{0}{2}\arrives\right)\mytag{(F7)} , & \metric{\next}{1}{1}\!\left(\moving \metric{\until}{0}{1}\arrives\right)\mytag{(F8)}, &
										\metric{\next}{1}{1}\!\left(\moving \metric{\until}{0}{0}\arrives\right) \mytag{(F9)}, \\
										\moving \metric{\until}{0}{2}\arrives \mytag{(F10)}, & \moving \metric{\until}{0}{1}\arrives \mytag{(F11)}, & \moving \metric{\until}{0}{0}\arrives,\\
										\moving, & \arrives , & \press
										\end{array}
										\right\}
\end{equation*}

Each formula in $\mathit{cl}(\eqref{ex:press:moving}) \setminus \lbrace \moving\metric{\until}{0}{0}\arrives,\moving,\press,\arrives \rbrace$ has been labelled in order to be referenced when displaying the translation of~\eqref{ex:press:moving} into a metric temporal logic program shown in Table~\ref{tbl:translation:example}.
Note that the formula $\mytag{(F1)}$ corresponds to~\eqref{ex:press:moving}.
The next proposition shows that the size of the closure of a formula $\varphi$, in symbols $\lvert \mathit{cl}(\varphi)\rvert$ can be bounded by a polynomial that depends on two parameters: the size of $\varphi$ (denoted by $\lvert \varphi \rvert$) and the maximum number $n$ used in a binary temporal operator occurring in $\varphi$.
\begin{proposition}\label{prop:closure:finite}
  For any metric formula $\varphi$, $cl(\varphi)$ is finite
  and
  the total size of all the formulas in $cl(\varphi)$, $\lvert cl(\varphi)\rvert$, is bounded by $\lvert cl(\varphi)\rvert
  \le 2(k_\varphi^2 + 1) \lvert \varphi \rvert$,
  where $k_\varphi = \max\{1,n_\varphi\}$ and
  $n_\varphi$ is the maximum $n$ that occurs in
  the operators $\metric{\until}{0}{n}$, $\metric{\release}{0}{n}$, $\metric{\since}{0}{n}$ and $\metric{\trigger}{0}{n}$ in $\varphi$.\qed
\end{proposition}
Note that the closure of a formula contains other formulas than subformulas,
which are also translated into logical rules.
To prove the correctness,
we cannot only rely on structural induction
since our translation depends on the notion of closure, which contains formulas that are not subformulas.
For instance, let us consider the formula $\moving \metric{\until}{0}{3} \arrives \in \mathit{cl}(\eqref{ex:press:moving})$.
Following Proposition~\ref{prop:alternative:semantics}, the formulas $\moving \metric{\until}{0}{3} \arrives$ and $\left(\arrives \vee \left(\moving \wedge \left( \metric{\next}{1}{1}\left(\moving \metric{\until}{0}{2} \arrives\right) \vee \metric{\next}{2}{2}\left(\moving \metric{\until}{0}{1} \arrives\right) \vee \metric{\next}{3}{3}  \left(\moving \metric{\until}{0}{0} \arrives\right)\right)\right)\right)$ are $\MHTf$-equivalent.
Our translation replaces
$\moving \metric{\until}{0}{3} \arrives$,
$\metric{\next}{1}{1}\left(\moving \metric{\until}{0}{2} \arrives\right)$,
$\metric{\next}{2}{2}\left(\moving \metric{\until}{0}{1} \arrives\right)$ and
$\metric{\next}{3}{3}  \left(\moving \metric{\until}{0}{0} \arrives\right)$ by fresh atoms and then
adds extra formulas that ensure the relationship between those and the formulas they replace.
Although the proof strategy is to proceed by structural induction, it is not directly possible here since, for instance,
$\metric{\next}{1}{1}\left(\moving \metric{\until}{0}{2} \arrives\right)$ is not a proper subformula of $\moving \metric{\until}{0}{3} \arrives$.
We aim to define a well-founded strict partial order that takes into account not only the complexity of the input formula but also the unfolding of the $\metric{\until}{0}{n}$, $\metric{\release}{0}{n}$, $\metric{\since}{0}{n}$ and $\metric{\trigger}{0}{n}$ operators. We associate with each metric temporal formula $\varphi$ a natural number $\dist(\varphi)$ that is used to define our well-founded strict partial order.
\begin{definition} For a given metric formula $\varphi$ in the language, we define $\dist(\varphi)$, that captures the unfolding of the formulas $\varphi\metric{\until}{0}{n} \psi$, $\varphi\metric{\release}{0}{n} \psi$, $\varphi\metric{\since}{0}{n} \psi$ and $\varphi\metric{\trigger}{0}{n} \psi$, as follows:
	\begin{itemize}
		\item $\dist(\bot) = \dist(\top) = 1$;
		\item $\dist(a) = 1$, for all $a \in \PV$;
		\item $\dist(\varphi \otimes \psi) = 1 + \dist(\varphi) + \dist(\psi)$, with $\otimes \in \lbrace \wedge, \vee, \rightarrow, \until, \release, \since,\trigger \rbrace$;
		\item $\dist(\oplus \varphi) = 1+\dist(\varphi)$, with $\oplus \in \lbrace\next, \previous, \wnext, \wprevious, \metricI{\next}, \metricI{\previous}, \metricI{\wnext}, \metricI{\wprevious}, \eventuallyF, \eventuallyP, \alwaysF, \alwaysP \rbrace$;
		\item $\dist(\varphi \metric{\otimes}{0}{n} \psi) = 2*(n+1)+\dist(\varphi) + \dist(\psi)$, with $\otimes \in \lbrace  \until, \release, \since,\trigger \rbrace$;
		\item $\dist(\metric{\oplus}{0}{n} \varphi)= 2*(n+1)+\dist(\varphi)$, with $\oplus \in \lbrace \eventuallyF, \eventuallyP, \alwaysF, \alwaysP \rbrace$.\qed
	\end{itemize}
\end{definition}
Given two metric formulas, we say that $\varphi \prec_{\dist} \psi$ if and only if $\dist(\varphi) < \dist(\psi)$. 
Since the previous definition is given as a mapping into the natural numbers, we can readily prove the following result.
\begin{proposition}\label{prop:wfspo}
	$\prec_{\dist}$ is a well-founded strict partial order.\qed
\end{proposition}
Regarding our running example, we get $\dist(\moving \metric{\until}{0}{3} \arrives) = 10$, $\dist(\metric{\next}{1}{1}\left(\moving \metric{\until}{0}{2} \arrives\right))= 9$,   $\dist(\metric{\next}{2}{2}\left(\moving \metric{\until}{0}{1} \arrives\right))= 7$ and $\dist(\metric{\next}{3}{3}  \left(\moving \metric{\until}{0}{0} \arrives\right))=5$.
Consequently, $\metric{\next}{1}{1}\left(\moving \metric{\until}{0}{2} \arrives\right) \prec_{\dist} \moving \metric{\until}{0}{3} \arrives$, $ \metric{\next}{2}{2}\left(\moving \metric{\until}{0}{1} \arrives\right) \prec_{\dist} \moving \metric{\until}{0}{3} \arrives$ and $ \metric{\next}{3}{3}  \left(\moving \metric{\until}{0}{0} \arrives\right) \prec_{\dist} \moving \metric{\until}{0}{3} \arrives$, which allows us to apply the induction on the unfolding of binary metric temporal operators.
Thanks to Proposition~\ref{prop:wfspo},
we can use $\dist$ to provide an induction hypothesis, which is not based on the typical structural induction.

The following theorem states that an arbitrary metric temporal formula $\varphi$\footnote{We assume that the binary metric temporal operators occurring in $\varphi$ are restricted to intervals of the form $[0..n]$.} over an alphabet $\PV$
can be compiled into a metric temporal logic program $\Gamma$ over an extended alphabet $\PV_{\Lab{}}\supseteq \PV$ in such a way
that there exists a bijection between the timed \HT{}-traces of length $\rangeco{\lambda}{1}{\omega}$ satisfying $\varphi$ and the traces of length $\lambda$ satisfying $\Gamma$ projected onto $\PV$.
That is, we have \emph{equivalence modulo auxiliary atoms}.
\begin{theorem}[Metric temporal logic program reduction]\label{thm:normalform:ht}
	For every metric temporal formula $\varphi$\footnotemark[\value{footnote}] over $\PV$ and for any $\rangeco{\lambda}{1}{\omega}$, we have that
	\begin{equation}
	 \MHT(\varphi,\lambda) = \lbrace  \M|_\mathcal{A} \mid  \M \in \MHT(\lbrace\Lab{\varphi} \rbrace \cup \eta^*(\varphi), \lambda) \rbrace.\label{ht:statement}
	\end{equation}
\qed
\end{theorem}
For transforming arbitrary temporal formulas into metric temporal logic programs,
we use a Tseitin-style reduction~\citep{tseitin68a} that relies on an alphabet extended
by new atoms for each formula in the original language.
The equivalence result in Theorem~\ref{thm:normalform:ht} is then obtained after removing auxiliary atoms and,
in fact, is still preserved inside the context of a larger theory for the original vocabulary.
The results presented in Theorem~\ref{thm:normalform:ht} are preserved when selecting the minimal models.
This means that there exists a bijection between the metric equilibrium models of $\varphi$ and $\Gamma$ projected onto $\PV$.
That is, we have \emph{strong equivalence modulo auxiliary atoms}.
\begin{corollary}\label{thm:normalform:eq}
	For every metric temporal formula $\varphi$\footnotemark[\value{footnote}] over $\PV$, for any $\rangeco{\lambda}{1}{\omega}$ and any theory $\Gamma$ over $\PV$,
	\begin{equation}
		\MEL(\lbrace \varphi \rbrace \cup \Gamma,\lambda) = \{\tuple{(\Htrace,\Ttrace),\tau}|_\mathcal{A} \mid \tuple{(\Htrace,\Ttrace),\tau} \in \MEL(\lbrace \Lab{\varphi} \rbrace \cup \eta^*(\varphi) \cup \Gamma,\lambda)\}.\label{eq:statement}
	\end{equation}\qed
\end{corollary}
Theorem~\ref{thm:normalform:ht} and Corollary~\ref{thm:normalform:eq} can be extended to arbitrary theories by defining the translation $\sigma$ as the metric temporal logic program:
\begin{align*}
	\sigma(\Gamma) =
	\left\lbrace \Lab{\gamma} \mid \gamma \in \Gamma \right\rbrace
	\cup \left\lbrace \eta^*(\mu)
	\mid \mu \in cl(\Gamma)\right\rbrace
\end{align*}
and then replacing $\lbrace \Lab{\varphi} \rbrace \cup \eta^*(\varphi)$ by $\sigma(\Gamma)$ in expressions~\eqref{ht:statement} and~\eqref{eq:statement}.
Table~\ref{tbl:translation:example} shows the translation of the expression in~\eqref{ex:press:moving},
which is done by applying the translation $\eta(\mu)$ for each $\mu\in \mathit{cl}(\eqref{ex:press:moving})$.
A metric temporal logic program can be obtained by using $\eta^*(\mu)$ instead.
\begin{table}[h!]\centering
	\caption{Translation of~\eqref{ex:press:moving} into a metric temporal logic program. Equivalences of the form $\Lab{\varphi}\leftrightarrow p$, with $p\in \PV$, are not shown and $p$ is used instead of $\Lab{p}$.}
	\label{tbl:translation:example}
		{\tablefont \begin{tabular}{@{\extracolsep{\fill}}cc}\topline
				$\mu\in \mathit{cl}(\eqref{ex:press:moving})$ & $\eta(\mu)$  \\\hline
				$\mytag{(F1)}$ & $\wnext \alwaysF\left(\previous \Lab{\mytag{(F1)}} \leftrightarrow \previous\Lab{\mytag{(F2)}} \wedge \Lab{\mytag{(F1)}}\right)$, $\alwaysF\left(\finally \to \left(\Lab{\mytag{(F1)}}\leftrightarrow \Lab{\mytag{(F2)}}  \right)\right)$\\\hline
				$\mytag{(F2)}$ &  $\alwaysF \left(\Lab{\mytag{(F2)}} \leftrightarrow \left(\press \to \Lab{\mytag{(F3)}}\right)\right)$\\\hline
				$\mytag{(F3)}$ &  $\alwaysF \left(\Lab{\mytag{(F3)}} \leftrightarrow a \vee \left( m \wedge \left(\Lab{\mytag{(F7)}} \vee \Lab{\mytag{(F5)}}\vee \Lab{\mytag{(F4)}} \right)\right)   \right)$  \\\hline
				$\mytag{(F4)}$ &  $\wnext \alwaysF \left(\metric{\previous}{3}{3}\top \to \left(\previous \Lab{\mytag{(F4)}} \leftrightarrow a\right) \right)$, $\wnext \alwaysF \left( \neg \metric{\previous}{3}{3} \top \rightarrow \neg \previous \Lab{\mytag{(F4)}} \right)$, $\alwaysF\left(\finally \to \neg \Lab{\mytag{(F4)}}\right)$  \\\hline
				$\mytag{(F5)}$ &    $\wnext \alwaysF \left(\metric{\previous}{2}{2}\top \to \left(\previous \Lab{\mytag{(F5)}} \leftrightarrow \Lab{\mytag{(F11)}}\right) \right)$, $\wnext \alwaysF \left( \neg \metric{\previous}{2}{2} \top \rightarrow \neg \previous \Lab{\mytag{(F5)}} \right)$, $\alwaysF\left(\finally \to \neg \Lab{\mytag{(F5)}}\right)$  \\\hline
				$\mytag{(F6)}$ &    $\wnext \alwaysF \left(\metric{\previous}{2}{2}\top \to \left(\previous \Lab{\mytag{(F6)}} \leftrightarrow a\right) \right)$, $\wnext \alwaysF \left( \neg \metric{\previous}{2}{2} \top \rightarrow \neg \previous \Lab{\mytag{(F6)}} \right)$, $\alwaysF\left(\finally \to \neg \Lab{\mytag{(F6)}}\right)$  \\\hline
				$\mytag{(F7)}$ &   $\wnext \alwaysF \left(\metric{\previous}{1}{1}\top \to \left(\previous \Lab{\mytag{(F7)}} \leftrightarrow \Lab{\mytag{(F10)}}\right) \right)$, $\wnext \alwaysF \left( \neg \metric{\previous}{1}{1} \top \rightarrow \neg \previous \Lab{\mytag{(F7)}} \right)$, $\alwaysF\left(\finally \to \neg \Lab{\mytag{(F7)}}\right)$  \\\hline
				$\mytag{(F8)}$ &    $\wnext \alwaysF \left(\metric{\previous}{1}{1}\top \to \left(\previous \Lab{\mytag{(F8)}} \leftrightarrow \Lab{\mytag{(F11)}}\right) \right)$, $\wnext \alwaysF \left( \neg \metric{\previous}{1}{1} \top \rightarrow \neg \previous \Lab{\mytag{(F8)}} \right)$, $\alwaysF\left(\finally \to \neg \Lab{\mytag{(F8)}}\right)$  \\\hline
				$\mytag{(F9)}$ &    $\wnext \alwaysF \left(\metric{\previous}{1}{1}\top \to \left(\previous \Lab{\mytag{(F9)}} \leftrightarrow a\right) \right)$, $\wnext \alwaysF \left( \neg \metric{\previous}{1}{1} \top \rightarrow \neg \previous \Lab{\mytag{(F9)}} \right)$, $\alwaysF\left(\finally \to \neg \Lab{\mytag{(F9)}}\right)$  \\\hline
				$\mytag{(F10)}$ &   $\alwaysF \left(\Lab{\mytag{(F10)}} \leftrightarrow a \vee \left( m \wedge \left(\Lab{\mytag{(F8)}} \vee \Lab{\mytag{(F6)}}\right)\right)   \right)$  \\\hline
				$\mytag{(F11)}$ &  $\alwaysF \left(\Lab{\mytag{(F11)}} \leftrightarrow a \vee \left( m \wedge \Lab{\mytag{(F9)}}\right)   \right)$  \\\hline
\end{tabular}}
\end{table}
The rest of this section is devoted to proving Theorem~\ref{thm:normalform:ht},
for which we extend Cabalar's reduction for the case \THTo~\citep{cabalar10a} to the case of~\MHTf.
The reduction uses an extended alphabet $\PV_{\Lab{}}\supseteq\mathcal{A}$
that additionally contains a new atom \Lab{\varphi} (aka label) for each formula $\varphi$ in the original language over $\mathcal{A}$.
For convenience, we use
$\Lab{\varphi} \overset{\mathit{def}}{=} \varphi $ if $\varphi$ is $\top , \bot$ or an atom $a \in \mathcal{A}$.
For any non-atomic formula $\gamma$ over $\mathcal{A}$,
we introduce the translation $\eta$ given along
Tables~\ref{tbl:translation:mnext-mprevious}, \ref{tbl:translation:msince-mtrigger}, and \ref{tbl:translation:muntil-mrelease},
and call $\eta(\gamma)$ the \emph{definition} of $\gamma$.
(Due to space constraints, we omit the well-known translation rules for Boolean and non-metric operators here;
these are provided in Tables~\ref{tbl:tseitin:prop} and~\ref{tbl:tseitin:ltl} of the supplementary material.)
This definition is not a logic program yet,
but it contains some double implication that can be easily transformed into
a metric temporal logic program format by simple, non-modal transformations in the propositional logic of \HT.
This translation of $\eta(\mu)$ into a metric temporal logic program is shown in the tables as $\eta^*(\mu)$.
In the rest of the section,
we indistinctly use $\eta(\mu)$ and $\eta^*(\mu)$, as the former is more suitable for developing the proofs.
As said before, their equivalence can be easily tested inside logic HT, and so, is not provided in this paper.

Given a trace $\M= \tuple{(\H,\T),\tau}$ of length $\lambda\in \mathbb{N}\cup \lbrace \omega \rbrace$, we define its \emph{projection} onto the alphabet $\mathcal{A}$ as
$\M\mid_{\mathcal{A}} = \tuple{(\H',\T'),\tau}$ where $H'_i = H_i \cap \mathcal{A}$ and $T'_i = T_i \cap \mathcal{A}$, for all $\rangeco{i}{0}{\lambda}$.
\begin{lemma}\label{lem:nf1}
	Let $\M = \tuple{(\H,\T),\lambda}$ be a \MHT\ model of length $\rangeco{\lambda}{1}{\omega}$ of a formula $\varphi$ over $\PV$.
	Then, there exists some timed \HT-trace $\M'=\tuple{(\H',\T'),\tau}$ over alphabet $\PV_{\Lab{}}$
	such that $\M = \M'|_{\PV}$ and $\M',0 \models \lbrace \Lab{\varphi} \rbrace\cup \eta^*(\varphi)$.\qed
\end{lemma}
\begin{lemma}\label{lem:nf2}
	Let $\varphi$ be a temporal formula over $\mathcal{A}$ and
	let $\tuple{(\H,\T),\tau}$ be a \MHT{} model, with length $\rangeco{\lambda}{1}{\omega}$, of $\lbrace \Lab{\varphi}\rbrace\cup \eta(\varphi)$.
    Let ${\bm m}$ be its associated three-valued interpretation.
	Then for any $\mu \in cl(\varphi)$ and any $\rangeco{k}{0}{\lambda}$,
	we have $\trival{k}{\Lab{\mu}} = \trival{k}{\mu}$.\qed
\end{lemma}
Let us discuss now the proof of Theorem~\ref{thm:normalform:ht}. 
We show that for any metric temporal formula $\varphi$ over $\mathcal{A}$ and any $\rangeco{\lambda}{1}{\omega}$,
\begin{equation*}
\MHT(\varphi,\lambda) = \lbrace  \M|_\mathcal{A} \mid  \M \in \MHT(\lbrace\Lab{\varphi} \rbrace \cup \eta(\varphi), \lambda) \rbrace.
\end{equation*}
The `$\subseteq$' direction follows from Lemma~\ref{lem:nf1}.
For the `$\supseteq$' direction, take some \MHT{} model $\M$ of $\eta(\varphi)$ of length $\rangeco{\lambda}{1}{\omega}$.
This implies that its associated $3$-valuation satisfies $\trivalp{0}{\Lab\varphi}=2$, which is included in $\eta(\varphi)$.
Since $\varphi \in cl(\varphi)$, we can apply Lemma~\ref{lem:nf2} to conclude $\trivalp{k}{\Lab\varphi}=\trivalp{k}{\varphi}$ for all $\rangeco{k}{0}{\lambda}$ but then, in particular, $2=\trivalp{0}{\Lab\varphi}=\trivalp{0}{\varphi}$ and so, $\M$ is also a model of $\varphi$.
Finally, $\M|_\mathcal{A}$ is still a model of $\varphi$ because the latter is a theory over vocabulary $\mathcal{A}$.
The theorem below provides a parametric polynomial bound of the size of $\eta^*(\varphi)$.
\begin{theorem}\label{them:polynomial}
The number of rules in $\eta^*(\varphi)$ has a (parametric) polynomial bound that depends on $n$ and $\lvert \varphi \rvert$,
  where $n$ is the maximum interval value occurring in a subformula of $\varphi$,
  which is of the form $\metric{Q}{0}{n}$, where $Q\in \lbrace \until,\release,\since,\trigger\rbrace$. \qed
\end{theorem}
 \section{Conclusion}

In this paper,
we have described a Tseitin-like reduction from metric temporal logic formulas
into a logic programming format under metric equilibrium logic semantics.
Our current translation only considers intervals of the form $\intervcc{0}{n}$
in the binary metric operators \emph{since} (resp.\ \emph{trigger}) and \emph{until} (resp.\ \emph{release}).
In the future, we expect to extend our translation in order to cope with arbitrary intervals.
If that is possible, 
we would have an input format to which any arbitrary metric temporal theory could be reduced, 
making it easily implementable in ASP.

Our contribution sheds light on the concept of metric temporal logic programs and
its use as a ``normal form'' for interpreting metric temporal theories under answer set semantics.
As future work, we expect to implement this translation in the solver \telingo~\citep{cakamosc19a} and
compare it with other solvers proposed by~\cite{becadiharosc24a,becadirohasc25a}.
In addition,
we would like to explore the notion of splittable metric temporal logic programs as done by~\cite{agcadipescscvi20a} for \TEL.
In the temporal case, those programs possess nice computational methods via the use of splitting and loop formulas.
An interesting research line could be to extend those results to the metric case.
Finally, we would like to extend our reduction to metric extensions of \emph{dynamic equilibrium logic}~\citep{becadifascsc24a}
 \section*{Acknowledgments}
We would like to sincerely thank the reviewers for their careful evaluation, insightful comments, and constructive suggestions, which greatly contributed to improving the quality and clarity of this work, as well as for identifying issues related to the size of the translation. We also extend our gratitude to David Fernández-Duque for his valuable advice, which helped us further improve the paper.
This work has been partially supported by DFG grant SCHA 550/15.

 \bibliographystyle{plainnat} 
 
\appendix
\clearpage{}
\section{Proofs (for Reviewing Purposes)}
\subsection{Proofs of Section~\ref{sec:approach}}

Satisfaction of derived operators can be easily deduced:
\begin{proposition}[\citealp{becadiscsc24a}]\label{prop:satisfaction:tel}
	Let $\M=(\tuple{\H,\T},\tmf)$
	be a timed \HT-trace of length $\lambda$ over \PV.
	Given the respective definitions of derived operators, we get the following satisfaction conditions:
	\begin{enumerate} \setcounter{enumi}{11}
		\item $\M, k \models \initially$
		iff
		$k =0$
		\item $\M, k \models \metricI{\wprevious}\, \varphi$
		iff
		$k =0$ or
		$\M, k{-}1 \models \varphi$
		or $\tmf(k)-\tmf(k{-}1) \not\in \cI$
		\item $\M, k \models \metricI{\eventuallyP}\, \varphi$
		iff
		$\M, i \models \varphi$ for some $\rangecc{i}{0}{k}$
		with
		$\tmf(k)-\tmf(i) \in \cI$
		\item $\M, k \models \metricI{\alwaysP}\, \varphi$
		iff
		$\M, i \models \varphi$ for all $\rangecc{i}{0}{k}$
		with
		$\tmf(k)-\tmf(i) \in \cI$
		\item $\M, k \models \finally$
		iff
		$k+1 = \lambda$
		\item \label{def:mhtsat:wnext} $\M, k \models \metricI{\wnext}\, \varphi$
		iff
		$k+1 = \lambda$
		or $\M, k{+}1 \models \varphi$
		or $\tmf(k{+}1)-\tmf(k) \not\in \cI$
		\item \label{def:mhtsat:eventuallyF} $\M, k \models \metricI{\eventuallyF}\, \varphi$
		iff
		$\M, i \models \varphi$ for some $\rangeco{i}{k}{\lambda}$
		with
		$\tmf(i)-\tmf(k) \in \cI$
		\item \label{def:mhtsat:alwaysF} $\M, k \models \metricI{\alwaysF}\, \varphi$
		iff
		$\M, i \models \varphi$ for all $\rangeco{i}{k}{\lambda}$
		with
		$\tmf(i)-\tmf(k) \in \cI$
		\qed
	\end{enumerate}
\end{proposition}

\begin{proofof}{Proposition~\ref{prop:mht:strict}}
		Let us take any timed trace $\M=(\tuple{\Htrace,\Ttrace},\tau)$.
		From left to right, if $\tau$ is not strict then there $\tau(i+1) -\tau(i) = 0$ for some $\rangeco{i}{0}{\lambda}$.
		Therefore, $\M, i \not \models \neg \metric{\next}{0}{0}\top$. Therefore, $\M,0 \not \models \alwaysF \neg \metric{\next}{0}{0} \top$.
		From right to left, if $\M, 0 \not \models \alwaysF \neg \metric{\next}{0}{0} \top$, $\M, i \not \models \neg \metric{\next}{0}{0} \top$ for some $\rangeco{i}{0}{\lambda}$.
		From the statement above it follows that $\M, i \models \metric{\next}{0}{0}\top$. From this we conclude that $i < \lambda-1$, otherwise the previous statement would not be valid.
		Therefore $\rangeco{i+1}{1}{\lambda}$ therefore $\tau(i+1)-\tau(i) = 0$: a contradiction. 
\end{proofof} \subsection{Proofs of Section~\ref{sec:three-valued}}
\begin{proposition}\label{prop:three-valued:derived}
	\allowdisplaybreaks The semantics for the derived operators is as follows:
	\begin{align*}
		\trival{k}{\metricI{\alwaysF} \varphi } & =  \inf \lbrace  \trival{i}{\varphi} \mid  \rangeco{i}{k}{\lambda} \text{ and } \tau(i)-\tau(k) \in I\rbrace \\
		\trival{k}{\alwaysF \varphi } & =  \min \lbrace  \trival{i}{\varphi} \mid  \rangeco{i}{k}{\lambda} \rbrace \\
		\trival{k}{\metricI{\eventuallyF} \varphi } & =  \sup \lbrace \trival{i}{\varphi} \mid  \rangeco{i}{k}{\lambda} \text{ and } \tau(i)-\tau(k) \in I\rbrace \\
		\trival{k}{\eventuallyF \varphi } & =  \max \lbrace \trival{i}{\varphi} \mid  \rangeco{i}{k}{\lambda} \rbrace \\
		\trival{k}{\metricI{\alwaysP} \varphi } & =  \inf \lbrace  \trival{i}{\varphi} \mid  \rangecc{i}{0}{k} \text{ and } \tau(k)-\tau(i) \in I\rbrace \\
		\trival{k}{\alwaysP \varphi } & =  \min \lbrace  \trival{i}{\varphi} \mid  \rangecc{i}{0}{k} \rbrace \\
		\trival{k}{\metricI{\eventuallyP} \varphi } & =  \sup \lbrace \trival{i}{\varphi} \mid  \rangecc{i}{0}{k} \text{ and } \tau(k)-\tau(i) \in I\rbrace \\
		\trival{k}{\eventuallyP \varphi } & =  \max \lbrace \trival{i}{\varphi} \mid  \rangecc{i}{0}{k} \rbrace \\
		\trival{k}{\finally} &= 2 \text{ if } k = \lambda-1\text{ and } 0 \text{ otherwise}\\
		\trival{k}{\initially} &= 2 \text{ if } k = 0\text{ and } 0 \text{ otherwise}\\
	\end{align*}
\end{proposition}

\begin{proofof}{Proposition~\ref{prop:alternative:semantics}}
We proceed by cases. 

\begin{itemize}
	\item Case $\trival{k}{\varphi \metric{\until}{0}{n}\psi }$: by the equivalence~\eqref{def:inductive_until_0n} of Proposition~\ref{lemma_inductive_def_interval_zero} we have
	\begin{eqnarray*}
	&&\trival{k}{\varphi \metric{\until}{0}{n}\psi } \\
	&=& \trival{k}{\psi \vee ( \varphi \wedge \bigvee_{x=1}^{n} \metric{\next}{x}{x} ( \varphi \metric{\until}{0}{(n-x)} \psi ) )}\\
	&=& \max\lbrace \trival{k}{\psi},\min\lbrace \trival{k}{\varphi}, \max\lbrace \trival{k}{\metric{\next}{x}{x} ( \varphi \metric{\until}{0}{(n-x)} \psi )} \mid \rangecc{x}{1}{n} \rbrace \rbrace \rbrace
	\end{eqnarray*}	
	\item Case $\varphi \metric{\release}{0}{n} \psi$: by the equivalence~\eqref{def:inductive_release_0n} of Proposition~\ref{lemma_inductive_def_interval_zero} we have 
		\begin{eqnarray*}
		&&\trival{k}{\varphi \metric{\release}{0}{n}\psi } \\
		&=& \trival{k}{\psi \wedge ( \varphi \vee \bigwedge_{x=1}^{n} \metric{\wnext}{x}{x} ( \varphi \metric{\release}{0}{(n-x)} \psi ) )}\\
		&=& \min\lbrace \trival{k}{\psi},\max\lbrace \trival{k}{\varphi}, \min\lbrace \trival{k}{\metric{\wnext}{x}{x} ( \varphi \metric{\release}{0}{(n-x)} \psi )} \mid \rangecc{x}{1}{n} \rbrace \rbrace \rbrace
		\end{eqnarray*}	
		
	\item Case $\varphi \metric{\since}{0}{n} \psi$: by the equivalence~\eqref{def:inductive_since_0n} of Proposition~\ref{lemma_inductive_def_interval_zero} we have 
	\begin{eqnarray*}
	&&\trival{k}{\varphi \metric{\since}{0}{n}\psi } \\
	&=& \trival{k}{\psi \vee ( \varphi \wedge \bigvee_{x=1}^{n} \metric{\previous}{x}{x} ( \varphi \metric{\since}{0}{(n-x)} \psi ) )}\\
	&=& \max\lbrace \trival{k}{\psi},\min\lbrace \trival{k}{\varphi}, \max\lbrace \trival{k}{\metric{\previous}{x}{x} ( \varphi \metric{\since}{0}{(n-x)} \psi )} \mid \rangecc{x}{1}{n} \rbrace \rbrace \rbrace
	\end{eqnarray*}	\item Case $\trival{k}{\varphi \metric{\trigger}{0}{n} \psi}$: by the equivalence~\eqref{def:inductive_trigger_0n} of Proposition~\ref{lemma_inductive_def_interval_zero} we have 

	\begin{eqnarray*}
	&&\trival{k}{\varphi \metric{\trigger}{0}{n}\psi } \\
	&=& \trival{k}{\psi \wedge ( \varphi \vee \bigwedge_{x=1}^{n} \metric{\wprevious}{x}{x} ( \varphi \metric{\trigger}{0}{(n-x)} \psi ) )}\\
	&=& \min\lbrace \trival{k}{\psi},\max\lbrace \trival{k}{\varphi}, \min\lbrace \trival{k}{\metric{\wprevious}{x}{x} ( \varphi \metric{\trigger}{0}{(n-x)} \psi )} \mid \rangecc{x}{1}{n} \rbrace \rbrace \rbrace
	\end{eqnarray*}	
\end{itemize}	
\end{proofof}

 \subsection{Equivalence of the Semantics}\label{sec:equiv-semantics}

\begin{proposition}\label{prop:bijective}
	$f$ is bijective.\qed
\end{proposition}
\begin{proof}We prove that $f$ is both, injective and surjective.
\paragraph{injective}: Let us take two timed \HT{}-traces $\M=(\tuple{\Htrace,\Ttrace},\tau)$ and $\M' = (\tuple{\Htrace',\Ttrace'},\tau')$ such that $\M'\not=\M$. If $\tau \not = \tau'$ then clearly $f(\M)=m(\tau) \not = m(\tau') = f(\M')$ by definition. Let us assume now that $\tau=\tau'$ and let us set $m(\tau)=f(\M)$ and and $m'(\tau)=f(\M')$.
\begin{enumerate}
	\item If $\Ttrace \not = \Ttrace'$ then there is $\rangeco{i}{0}{\lambda}$ such that $T_i\not = T'_i$. Assume, without loss of generality that $p\in T_i$ but $p \not \in T'_i$, for some $p\in \PV$. 
	By definition, $m(\tau)(k,p) \ge 1$ and $m'(\tau)(k,p) = 0$ meaning that $f(\M)\not=f(\M')$. 
	\item If $\Htrace \not = \Htrace'$  then there is $\rangeco{i}{0}{\lambda}$ such that $H_i\not = H'_i$. Assume, without loss of generality that $p\in H_i$ but $p \not \in H'_i$, for some $p\in \PV$. 
	By definition, $m(\tau)(k,p) =2$ and $m'(\tau)(k,p) < 2$ meaning that $f(\M)\not=f(\M')$.
\end{enumerate}

\paragraph{surjective}: take any $m(\tau)$ be a three-valued interpretation. Let us define $\M=(\tuple{\Htrace,\Ttrace},\tau)$ as follows: 
\begin{equation*}
H_k=\lbrace p \mid p \in \PV \hbox{ and } m(k,p)=2\rbrace \hspace{30pt}	T_k=\lbrace p \mid p \in \PV \hbox{ and } m(k,p)\not= 0\rbrace
\end{equation*}
\noindent for $\rangeco{k}{0}{\lambda}$. it can be checked that $f(\M) = m(\tau)$.
\end{proof}	

\begin{proofof}{Proposition~\ref{prop:equivalence:semantics}}
By structural induction on $\varphi$. 
For simplicity, we will take $m(\tau)=f(\M)$ and we will denote by $\bm$ the extension of $m(\tau)$ to the metric temporal language.
We only present the case of the until operator.

From left to right, let us assume that $\M, k \models \varphi \metricI{\until} \psi$.
This means that there exists $\rangeco{i}{k}{\lambda}$ such that $\M, i \models \psi$, $\tau(i)-\tau(k)\in I$ and 
$\M, j \models \varphi$, for all $\rangeco{j}{k}{i}$.
By induction, $\trival{i}{\psi}=2$ and $\trival{j}{\varphi}=2$, for all $\rangeco{j}{k}{i}$.
As a consequence, $\min \lbrace \trival{i}{\psi}, \trival{j}{\varphi}  \mid \rangeco{j}{k}{i}\rbrace = 2$
Since $\tau(i)-\tau(k)\in I$ we get that $\sup \lbrace \min \lbrace \trival{i}{\psi}, \trival{j}{\varphi}  \mid \rangeco{j}{k}{i}\rbrace \mid \tau(i)-\tau(k)\in I \rbrace = 2$. 

Conversely, assume that $\trival{k}{\varphi \metricI{\until} \psi} = 2$. This means that 
\begin{equation*}
\sup\{\min\{\trival{i}{\psi},\trival{j}{\varphi}\mid \rangeco{j}{k}{i}\}\mid \rangeco{i}{k}{\lambda} \hbox{ and } \tau(i)-\tau(k)\in I\}=2.
\end{equation*}
Since $\sup\emptyset = 0$,
\begin{equation*}
\lbrace \min\{\trival{i}{\psi},\trival{j}{\varphi}\mid \rangeco{j}{k}{i}\}\mid \rangeco{i}{k}{\lambda} \hbox{ and } \tau(i)-\tau(k)\in I\rbrace \not= \emptyset.	
\end{equation*}
Therefore, $\min\{\trival{i}{\psi},\trival{j}{\varphi}\mid \rangeco{j}{k}{i}\}=2$ for some $\rangeco{i}{k}{\lambda}$ satisfying $\tau(i)-\tau(k)\in I$.
From such previous assumption, it we obtain $\trival{i}{\psi}= 2$, $\trival{j}{\varphi}=2$ for some $\rangeco{i}{k}{\lambda}$ satisfying $\tau(i)-\tau(k)\in I$ and for every $\rangeco{j}{k}{i}$.
By induction we have $\M, i \models \psi$ and $\M, j \models \varphi$, for all $\rangeco{j}{k}{i}$. Together with $\tau(i)-\tau(k)\in I$ and $\rangeco{i}{k}{\lambda}$ we get $\M, k \models \varphi \metricI{\until}\psi$.
The second item can be proved similarly.
\end{proofof}
 \subsection{Proofs of Section~\ref{sec:translation}}
\begin{proofof}{Proposition~\ref{prop:closure:finite}}
	We assume $k_\varphi$ is defined for any metric formula $\varphi$ as stated in the proposition.
	We proceed by structural induction.
	\begin{itemize}
		\item case $\varphi = p$: in this case $\lvert p \rvert = 1$ and $ cl(p) = \lbrace p \rbrace$, so $\lvert  cl(p)\rvert = 1$. 
		Since there is no metric connective in $p$, $k_p = 1$. Therefore, $\lvert  cl(p)\rvert = 1 \le 2(1^2 + 1) \lvert p \rvert = 4$.
		\item case $\varphi \wedge \psi$: in this case, $ cl(\varphi \wedge \psi) = \lbrace \varphi \wedge \psi \rbrace \cup  cl(\varphi) \cup  cl(\psi)$ and $\lvert cl(\varphi \wedge \psi) \rvert \le 1 + \lvert  cl(\varphi) \rvert + \lvert  cl(\psi)\rvert$.
		We assume, without loss of generality, $\max(k_\varphi,k_\psi) = k_\psi$ when needed. 
		We prove this case below:
		
		\begin{align*}
			&\lvert  cl(\varphi \wedge \psi)\rvert\\ 
			&\le 1 + \lvert  cl(\varphi) \rvert + \lvert  cl(\psi)\rvert & \\
			& \le 1 + 2(k_\varphi^2 +1)\lvert \varphi\rvert + 2(k_\psi^2 +1)\lvert \psi\rvert & \hbox{by induction on }\varphi \hbox{ and } \psi\\
			& \le 1 + 2k_\varphi^2\lvert \varphi\rvert  + 2\lvert \varphi\rvert+ 2k_\psi^2\lvert \psi\rvert + 2\lvert \psi\rvert \\
			& \le 1 + 2k_\varphi^2\lvert \varphi\rvert  + 2\lvert \varphi\rvert+ 2k_\psi^2\lvert \psi\rvert + 2\lvert \psi\rvert + \underbrace{2(k_\psi^2 - k_\varphi^2)\lvert \varphi \rvert}_{2(k_\psi^2 - k_\varphi^2)\lvert \varphi \rvert \ge 0} \\
			& \le 1 + \xcancel{2k_\varphi^2\lvert \varphi\rvert}  + 2\lvert \varphi\rvert+ 2k_\psi^2\lvert \psi\rvert + 2\lvert \psi\rvert + 2k_\psi^2\lvert \varphi \rvert - \xcancel{2k_\varphi^2\lvert \varphi \rvert} \\
			& \le 1 +  2k_\psi^2\left(\lvert \varphi\rvert + \lvert \psi\rvert\right) + 2\left(\lvert \varphi\rvert + \lvert \psi\rvert\right)\\
			& \le 1 +  2k_\psi^2\left(\lvert \varphi \wedge \psi \rvert -1\right) + 2\left(\lvert \varphi \wedge \psi \rvert -1\right)\\
			& \le 1 +  2k_\psi^2\lvert \varphi \wedge \psi \rvert - 2k_\psi^2 + 2\lvert \varphi \wedge \psi \rvert - 2\\
			& \le 2(k_\psi^2 + 1)\lvert \varphi \wedge \psi \rvert \underbrace{- 2k_\psi^2 -1}_{- 2k_\psi^2 -1 \le  0} \\
			& \le 2(k_\psi^2 + 1)\lvert \varphi \wedge \psi \rvert												
		\end{align*}			
		\item the proof for the formulas $\varphi \vee \psi$ and $\varphi \rightarrow \psi$ is done as for $\varphi \wedge \psi$.
		\item case $\metricI{\next} \varphi$: in this case, $cl(\metricI{\next} \varphi) = \lbrace \metricI{\next} \varphi \rbrace \cup  cl(\varphi)$ and $\lvert \metricI{\next} \varphi \rvert = 1 + \lvert \varphi \rvert$.
		It follows that 		
		\begin{align*}
			\lvert  cl(\metricI{\next} \varphi)\rvert & \le 1 + \lvert  cl(\varphi) \rvert\\
			& \le 1 +  2(k_\varphi^2+1) \lvert \varphi \rvert & \hbox{by induction on }\varphi\\
			& \le 1 +  2(k_\varphi^2+1) \left(\lvert \metricI{\next}\varphi \rvert - 1\right) \\
			& \le 1 +  2(k_\varphi^2+1)\lvert \metricI{\next}\varphi \rvert - 2(k_\varphi^2+1) \\
			& \le 2(k_\varphi^2+1)\lvert \metricI{\next}\varphi \rvert \underbrace{- 2k_\varphi^2-1}_{- 2k_\varphi^2-1 \le 0}\\
			& \le 2(k_\varphi^2+1)\lvert \metricI{\next}\varphi \rvert
		\end{align*}	
		\item  The proof for $\metricI{\wnext}\varphi$, $\metricI{\previous} \varphi$ and $\metricI{\wprevious} \varphi$ follows the same line of reasoning as for $\metricI{\next} \varphi$.
			
		\item Case $\varphi \metric{\until}{0}{n} \psi$: in this case we have that $\lvert\varphi \metric{\until}{0}{n} \psi \rvert = 1 + \lvert \varphi \rvert + \lvert \psi \rvert$ and 		
		\begin{equation*}
 			cl(\varphi \metric{\until}{0}{n} \psi) = cl(\varphi) \cup cl(\psi )\cup \lbrace  \varphi \metric{\until}{0}{i}\psi \mid \rangecc{i}{0}{n}\rbrace \cup \bigcup\limits_{i=1}^n\lbrace \metric{\next}{i}{i}\left(\varphi \metric{\until}{0}{j} \psi \right) \mid \rangecc{j}{0}{n-i}\rbrace.
		\end{equation*}		
	     \noindent Therefore, $\lvert cl(\varphi \metric{\until}{0}{n} \psi) \rvert \le  \frac{n^2 + 3n + 2}{2} +  \lvert  cl(\varphi) \rvert + \lvert  cl(\psi)\rvert \le 2(n^2+1) +  \lvert  cl(\varphi) \rvert + \lvert  cl(\psi)\rvert$.		 
		We consider two different cases:
		\begin{enumerate}
			\item $n = \max\lbrace n, k_\varphi,k_\psi \rbrace$. Consequently, $2(n^2 - k_\varphi^2) \lvert \varphi \rvert \ge 0$ and $2(n^2 - k_\psi^2) \lvert \psi \rvert\ge 0$:
			\begin{align*} 
				&\lvert  cl(\varphi \metric{\until}{0}{n} \psi)\rvert\\ 
				& \le 2(n^2 + 1) + \lvert  cl(\varphi) \rvert + \lvert  cl(\psi)\rvert \\
				& \le  2(n^2 + 1) + 2(k_\varphi^2 + 1)\lvert \varphi\rvert + 2(k_\psi^2 + 1)\lvert \psi\rvert \hspace{20pt} \hbox{by induction on }\varphi \hbox{ and } \psi \\
				& \le  2(n^2 + 1) + 2k_\varphi^2\lvert \varphi\rvert + 2\lvert \varphi\rvert + 2k_\psi^2\lvert \psi\rvert + 2\lvert \psi\rvert \\	
				& \le  2(n^2 + 1) + 2k_\varphi^2\lvert \varphi\rvert + 2\left( \lvert \varphi\rvert +\lvert \psi\rvert \right)+ 2k_\psi^2\lvert \psi\rvert \\	
				& \le  2(n^2 + 1) + 2k_\varphi^2\lvert \varphi\rvert + 2\left( \lvert \varphi\rvert +\lvert \psi\rvert \right)+ 2k_\psi^2\lvert \psi\rvert + \underbrace{2(n^2 - k_\varphi^2) \lvert \varphi \rvert}_{2(n^2 - k_\varphi^2) \lvert \varphi \rvert \ge 0} \\	
				& \le  2(n^2 + 1) + \xcancel{2k_\varphi^2\lvert \varphi\rvert} + 2\left( \lvert \varphi\rvert +\lvert \psi\rvert \right)+ 2k_\psi^2\lvert \psi\rvert + 2n^2 \lvert \varphi \rvert  \xcancel{-2k_\varphi^2\lvert \varphi \rvert}  \\			
				& \le  2(n^2 + 1) +  2\left( \lvert \varphi\rvert +\lvert \psi\rvert \right)+ 2k_\psi^2\lvert \psi\rvert + 2n^2 \lvert \varphi \rvert \\
				& \le  2(n^2 + 1) +  2\left( \lvert \varphi\rvert +\lvert \psi\rvert \right)+ 2k_\psi^2\lvert \psi\rvert + 2n^2 \lvert \varphi \rvert + \underbrace{2(n^2 - k_\psi^2) \lvert \psi \rvert}_{2(n^2 - k_\psi^2) \lvert \psi \rvert\ge 0} \\	
				& \le  2(n^2 + 1) +  2\left( \lvert \varphi\rvert +\lvert \psi\rvert \right)+ \xcancel{2k_\psi^2\lvert \psi\rvert} + 2n^2 \lvert \varphi \rvert + 2n^2 \lvert \psi \rvert \xcancel{- 2k_\psi^2\lvert \psi \rvert} \\		
				& \le  2(n^2 + 1) +  2\left( \lvert \varphi\rvert +\lvert \psi\rvert \right)+ 2n^2 \left( \lvert \varphi \rvert + \lvert \psi \rvert \right)  \\	
				& \le  2(n^2 + 1) +  2\left( \lvert \varphi \metric{\until}{0}{n}\psi\rvert-1\right)+ 2n^2 \left( \lvert \varphi \metric{\until}{0}{n}\psi\rvert-1\right) 
				 \\																							
				& \le  2(n^2 + 1) +  2(n^2 + 1)\left( \lvert \varphi \metric{\until}{0}{n}\psi\rvert-1\right) \\
				& \le  \xcancel{2(n^2 + 1)} +  2(n^2 + 1)\lvert \varphi \metric{\until}{0}{n}\psi\rvert \xcancel{- 2(n^2 + 1)} \\	
				& \le 2(n^2 + 1)\lvert \varphi \metric{\until}{0}{n}\psi\rvert  \\																																																																																																																																																																																																																																																																																																																																																																																																																																																																																																																																																																																																																																																																															
			\end{align*}	
			
			\item otherwise, assume that $k_\psi  = \max\lbrace k_\varphi,k_\psi, n \rbrace$. Therefore, $2(k_\psi^2 - k_\varphi^2) \lvert \varphi \rvert \ge 0$ :		
			\begin{align*} 
				&\lvert  cl(\varphi \metric{\until}{0}{n} \psi)\rvert \\
				& \le 2(n^2+1) + \lvert  cl(\varphi) \rvert + \lvert  cl(\psi)\rvert \\
				& \le 2(n^2+1)+ 2(k_\varphi^2+1) \lvert \varphi\rvert + 2(k_\psi^2+1)\lvert \psi\rvert \hspace{20pt} \hbox{by induction on }\varphi \hbox{ and } \psi \\
				& \le 2(n^2+1)+ 2k_\varphi^2\lvert \varphi\rvert + 2\lvert \varphi\rvert + 2k_\psi^2\lvert \psi\rvert + 2\lvert \psi\rvert\\ 
				& \le 2(n^2+1)+ 2k_\varphi^2\lvert \varphi\rvert  + 2k_\psi^2\lvert \psi\rvert + 2\left(\lvert \varphi\rvert + \lvert \psi\rvert\right)\\ 	
				& \le 2(n^2+1)+ 2k_\varphi^2\lvert \varphi\rvert  + 2k_\psi^2\lvert \psi\rvert + 2\left(\lvert \varphi\rvert + \lvert \psi\rvert\right) + \underbrace{2(k_\psi^2 - k_\varphi^2) \lvert \varphi \rvert}_{2(k_\psi^2 - k_\varphi^2) \lvert \varphi \rvert \ge 0}\\ 
				& \le 2(n^2+1)+ \xcancel{2k_\varphi^2\lvert \varphi\rvert}  + 2k_\psi^2\lvert \psi\rvert + 2\left(\lvert \varphi\rvert + \lvert \psi\rvert\right) + 2k_\psi^2\lvert \varphi \rvert \xcancel{- 2k_\varphi^2\lvert \varphi \rvert} \\ 
				& \le 2(n^2+1)+ 2k_\psi^2 \left( \lvert \varphi\rvert + \lvert \psi\rvert\right) + 2\left(\lvert \varphi\rvert + \lvert \psi\rvert\right)   \\ 	
				& \le 2(n^2+1)+ 2\left(k_\psi^2 + 1\right) \left( \lvert \varphi\rvert + \lvert \psi\rvert\right)   \\ 
				& \le 2(n^2+1)+ 2\left(k_\psi^2 + 1\right) \left( \lvert \varphi \metric{\until}{0}{n}\psi\rvert -1\right)\\  
				& \le 2(n^2+1)+ 2\left(k_\psi^2 + 1\right)  \lvert \varphi \metric{\until}{0}{n}\psi \rvert -2\left(k_\psi^2 + 1\right)  \\ 		
				& \le  2\left(k_\psi^2 + 1\right)  \lvert \varphi \metric{\until}{0}{n}\psi \rvert + \underbrace{2n^2 -2k_\psi^2}_{2n^2 -2k_\psi^2 \le 0}   \\
				& \le 2\left(k_\psi^2 + 1\right)  \lvert \varphi \metric{\until}{0}{n}\psi \rvert 
			\end{align*}	
		\end{enumerate}
		\item The proof for $\varphi \metric{\release}{0}{n} \psi$, $\varphi \metric{\since}{0}{n} \psi$ and $\varphi \metric{\trigger}{0}{n} \psi$ follow the same line of reasoning as for $\varphi \metric{\until}{0}{n} \psi$.
	\end{itemize}
\end{proofof}

\begin{proofof}{Lemma~\ref{lem:nf1}}
	Take the timed \HT-trace $(\tuple{\H',\T'},\tau)$ whose three valued interpretation ${\bm m}'$ satisfies:
	\begin{align}
		\trivalp{k}{\Lab{\varphi}} = \trival{k}{\varphi} \label{proofNF_equalityExtAlph}
	\end{align}
	for any formula $\varphi$ over $\mathcal{A}$ and for all $\rangeco{k}{0}{\lambda}$.
	When $\varphi$ is an atom $a \in \mathcal{A}$ then $\trivalp{k}{a} = \trivalp{k}{\Lab{a}} = \trival{k}{a}$, which implies that both valuations coincide for atoms, and so, $(\tuple{\H',\T'},\tau) |_{\mathcal{A}} = (\tuple{\H,\T},\tau)$.
	It remains to be show that
	$(\tuple{\H',\T'},\tau),0 \models \lbrace \Lab{\varphi}\rbrace \cup \eta(\varphi)$,
	which is equivalent to
	\begin{align*}
		(\tuple{\H',\T'}, \tau),0&\models
		\left\lbrace \Lab{\varphi} \right\rbrace
		\cup \left\lbrace \eta(\mu)
		\mid \mu \in cl(\varphi)\right\rbrace\\
		\Leftrightarrow (\tuple{\H',\T'},\tau),0 &\models
		\left\lbrace \Lab{\varphi} \right\rbrace
		\text{ and }
		(\tuple{\H',\T'},\tau),0 \models \left\lbrace \eta(\mu)
		\mid \mu \in cl(\varphi)\right\rbrace
	\end{align*}
	
	The first satisfaction relation follows directly from the definition of $(\tuple{\H',\T'}, \tau)$ since $\trivalp{0}{\Lab{\varphi}}=2$ iff $\trival{0}{\varphi}=2$ and we had that $(\tuple{\H,\T},0)$ is a model of $\varphi$.
	For the second part, we consider the following cases, depending on the $\mu$:
	
	\begin{enumerate}
		\item For $\mu = \varphi \otimes \psi$ with $\otimes \in \{\wedge, \vee, \to\}$ we have $\eta(\mu) = \alwaysF (\Lab{\mu} \leftrightarrow \Lab{\varphi} \otimes \Lab{\psi})$ and so, $(\tuple{\H',\T'},\tau),0 \models \eta(\mu)$ amounts to proving $\trivalp{k}{\Lab{\mu}}=\trivalp{k}{\Lab{\varphi} \otimes \Lab{\psi}}$ for all $\rangeco{k}{0}{\lambda}$.\\
		In this case we have
		\begin{align*}
			\trivalp{k}{\Lab{\mu}} = \trival{k}{\mu}
			&= \trival{k}{\varphi\otimes\psi}\\
			&= f^\otimes (\trival{k}{\varphi}, \trival{k}{\psi})\\
			&= f^\otimes (\trivalp{k}{\Lab{\varphi}}, \trivalp{k}{\Lab{\psi}})\\
			&= \trivalp{k}{\Lab{\varphi} \otimes \Lab{\psi} }
		\end{align*}
		where $f^\otimes$ is the three-valued mapping associated the corresponding propositional connective $\otimes \in \{\wedge, \vee, \to\}$ in the \THT{} three-valued semantics.
		\item Cases for non-metric temporal connectives are proved as in~\cite{agcadipescscvi20a}.
		
		\item For $\mu = \metricI{\previous} \varphi$ we have two formulas in $\eta(\mu)$:
		\begin{itemize}
			\item[-] For the first formula, $\wnext\alwaysF\left(\Lab{\mu} \leftrightarrow  \metricI{\previous}\Lab{\varphi}\right)$ we have to prove that 
			 $\trivalp{k}{\Lab{\mu}} = \trivalp{k}{\metricI{\previous}\Lab{\varphi}}$ for all $\rangeco{k}{1}{\lambda}$. Moreover, that this is trivially true when $\lambda=1$.
			 We reason by cases as follows:
			 \begin{enumerate}
			 	\item If $ \tau(k)-\tau(k-1) \in I$ then we have 
			 	\begin{align*}
			 		\trivalp{k}{\Lab{\mu}} = \trival{k}{\mu}
			 		&= \trival{k}{\metricI{\previous} \varphi}\\
			 		&= \trival{k-1}{\varphi} & \hbox{ since } \tau(k)-\tau(k-1) \in I\\
			 		&= \trivalp{k-1}{\Lab\varphi}\\
			 		&= \trivalp{k}{\metricI{\previous} \Lab{\varphi} }. 
			 	\end{align*}
			 	\item If $ \tau(k)-\tau(k-1) \not \in I$ then we have 
			 	\begin{align*}
			 		\trivalp{k}{\Lab{\mu}} = \trival{k}{\mu}
			 		&= \trival{k}{\metricI{\previous} \varphi}\\
			 		&= 0 & \hbox{ since } \tau(k)-\tau(k-1) \not \in I\\
			 		&= \trivalp{k}{\metricI{\previous} \Lab{\varphi}}. 
			 	\end{align*}
			 \end{enumerate}
			\item[-] For satisfying formula, $\neg \Lab\mu$, 
			amounts to check $\trivalp{0}{\Lab\mu}=0$ and this follows from 
			$\trivalp{0}{\Lab\mu}=\trival{0}{\mu}=\trival{0}{\metricI{\previous} \varphi}=0$.
		\end{itemize}

		\item For $\mu = \metricI{\wprevious} \varphi$ we have two formulas in $\eta(\mu)$:
		
		\begin{itemize}
			\item[-] For the first formula, $\wnext\alwaysF\left(\Lab{\mu} \leftrightarrow  \metricI{\wprevious}\Lab{\varphi}\right)$ we have to prove that 
			$\trivalp{k}{\Lab{\mu}} = \trivalp{k}{\metricI{\wprevious}\Lab{\varphi}}$ for all $\rangeco{k}{1}{\lambda}$. Moreover, that this is trivially true when $\lambda=1$.
			We reason by cases as follows:
			\begin{enumerate}
				\item If $ \tau(k)-\tau(k-1) \in I$ then we have 
				\begin{align*}
					\trivalp{k}{\Lab{\mu}} = \trival{k}{\mu}
					&= \trival{k}{\metricI{\wprevious} \varphi}\\
					&= \trival{k-1}{\varphi} & \hbox{ since } \tau(k)-\tau(k-1) \in I\\
					&= \trivalp{k-1}{\Lab\varphi}\\
					&= \trivalp{k}{\metricI{\wprevious} \Lab{\varphi} }. 
				\end{align*}
				\item If $ \tau(k)-\tau(k-1) \not \in I$ then we have 
				\begin{align*}
					\trivalp{k}{\Lab{\mu}} = \trival{k}{\mu}
					&= \trival{k}{\metricI{\wprevious} \varphi}\\
					&= 2 & \hbox{ since } \tau(k)-\tau(k-1) \not \in I\\
					&= \trivalp{k}{\metricI{\wprevious} \Lab{\varphi}}. 
				\end{align*}
			\end{enumerate}
			\item[-] For satisfying formula, $\Lab\mu$, 
			amounts to check $\trivalp{0}{\Lab\mu}=2$ and this follows from 
			$\trivalp{0}{\Lab\mu}=\trival{0}{\mu}=\trival{0}{\metricI{\wprevious} \varphi}=2$.
		\end{itemize}
	\item For $\mu = \varphi \metric{\since}{0}{n} \psi$, we have one formula in $\eta(\mu)$. In this case, we have to prove that 
	\begin{equation*}
		\trivalp{k}{\Lab{\mu}}=\trivalp{k}{\Lab{\psi} \vee \left( \Lab{\varphi} \wedge \bigvee\limits_{x=1}^n \Lab{\metric{\previous}{x}{x} (\varphi\metric{\since}{0}{(n-x)} \psi)}\right)} \\
	\end{equation*}	
	\noindent for all $\rangeco{k}{0}{\lambda}$ and that is proved next:
	\begin{eqnarray*}
		&&\trivalp{k}{\Lab{\mu}} \\
		&=&\trival{k}{\mu}\\
		&=&\trival{k}{\varphi \metric{\since}{0}{n} \psi}\\
		&\stackrel{Prop~\ref{prop:alternative:semantics}}{=}& \max\lbrace \trival{k}{\psi},\min\lbrace  \trival{k}{\varphi},\max\lbrace \trival{k}{\metric{\previous}{x}{x} (\varphi\metric{\since}{0}{n-x} \psi)}\mid \rangecc{x}{1}{n}  \rbrace  \rbrace \rbrace\\
		&=& \max\lbrace \trivalp{k}{\Lab\psi},\min\lbrace  \trivalp{k}{\Lab\varphi},\max\lbrace \trivalp{k}{\Lab{\metric{\previous}{x}{x} (\varphi\metric{\since}{0}{n-x} \psi)}}\mid \rangecc{x}{1}{n}  \rbrace  \rbrace \rbrace\\
		&=&\trivalp{k}{\Lab{\psi} \vee \left( \Lab{\varphi} \wedge \bigvee\limits_{x=1}^n \Lab{\metric{\previous}{x}{x} (\varphi\metric{\since}{0}{(n-x)} \psi)}\right)}
	\end{eqnarray*}
			
		\item For $\mu = \varphi \metric{\trigger}{0}{n} \psi$, we have only one formula in $\eta(\mu)$. In this case, we have to prove that 
		\begin{equation*}
			\trivalp{k}{\Lab{\mu}}=\trivalp{k}{\Lab{\psi} \wedge \left( \Lab{\varphi} \vee \bigwedge\limits_{x=1}^n \Lab{\metric{\wprevious}{x}{x} (\varphi\metric{\trigger}{0}{(n-x)} \psi)}\right)} 
		\end{equation*}	
		\noindent for all $\rangeco{k}{0}{\lambda}$ and that is proved next:
		\begin{eqnarray*}
			&&\trivalp{k}{\Lab{\mu}} \\
			&=&\trival{k}{\mu}\\
			&=&\trival{k}{\varphi \metric{\trigger}{0}{n} \psi}\\
			&\stackrel{Prop~\ref{prop:alternative:semantics}}{=}& \min\lbrace \trival{k}{\psi},\max\lbrace  \trival{k}{\varphi},\min\lbrace \trival{k}{\metric{\wprevious}{x}{x} (\varphi\metric{\trigger}{0}{n-x} \psi)}\mid \rangecc{x}{1}{n}  \rbrace  \rbrace \rbrace\\
			&=& \min\lbrace \trivalp{k}{\Lab\psi},\max\lbrace  \trivalp{k}{\Lab\varphi},\min\lbrace \trivalp{k}{\Lab{\metric{\wprevious}{x}{x} (\varphi\metric{\trigger}{0}{n-x} \psi)}}\mid \rangecc{x}{1}{n}  \rbrace  \rbrace \rbrace\\
			&=&\trivalp{k}{\Lab{\psi} \wedge \left( \Lab{\varphi} \vee \bigwedge\limits_{x=1}^n \Lab{\metric{\wprevious}{x}{x} (\varphi\metric{\trigger}{0}{(n-x)} \psi)}\right)}.
		\end{eqnarray*}
		
		\item For $\mu = \metricI{\next} \varphi$ we have three formulas in $\eta(\mu)$:		
		\begin{itemize}
			\item[-] For the first formula, $\wnext\alwaysF\left( \metricI{\previous} \top \to \left(\previous \Lab{\mu} \leftrightarrow \Lab{\varphi}\right) \right)$, it amounts to prove that 
			$\trivalp{k}{\Lab{\varphi}} = \trivalp{k}{\previous \Lab{\mu}}$, for all $\rangeco{k}{1}{\lambda}$ satisfying $\tau(k)-\tau(k-1) \in I$ (being trivially true when $\lambda=1$). 
			We reason as follows:
			\begin{align*}
				\trivalp{k}{\previous \Lab{\mu}} &= \trivalp{k-1}{\Lab\mu} &\\
				&= \trival{k-1}{\mu}\\
				&= \trival{k-1}{\metricI{\next}\varphi}\\
				&= \trival{k}{\varphi} & \hbox{since } \tau(k)-\tau(k-1) \in I\\
				&= \trivalp{k}{\Lab{\varphi} }
			\end{align*}
			\item[-] For the second formula, $\wnext\alwaysF\left( \neg \metricI{\previous} \top \to \previous \neg \Lab{\mu} \right)$, it amounts to prove that 
			$\trivalp{k}{\previous \neg \Lab{\mu}} = 2$, for all $\rangeco{k}{1}{\lambda}$ satisfying $\tau(k)-\tau(k-1) \not \in I$.
			This is equivalent to show that $\trivalp{k-1}{\Lab{\mu}} = 0$, for all $\rangeco{k}{1}{\lambda}$ satisfying $\tau(k)-\tau(k-1) \not \in I$, which is done next:
			\begin{equation*}
					\trivalp{k-1}{\Lab{\mu}} = \trival{k-1}{\mu} 
					= \trival{k-1}{\metricI{\next}\varphi}
					= 0, \hspace{10pt} \hbox{since } \tau(k)-\tau(k-1) \not \in I.
			\end{equation*}
			\item[-] For the third formula, $\alwaysF (\finally \to \neg \Lab{\mu})$, it amounts to show $\trivalp{\lambda-1}{\neg\Lab{\mu}}=2$.
			The latter is is equivalent to $\trivalp{\lambda-1}{\Lab{\mu}}=0$ which follows by construction of ${\bm m}'$, since $\trivalp{\lambda-1}{\Lab{\mu}}=\trival{\lambda-1}{\mu}=\trival{\lambda -1}{\metricI{\next}\varphi}=0$.
		\end{itemize}
		\item For $\mu = \metricI{\wnext} \varphi$ we have three formulas in $\eta(\mu)$:		
\begin{itemize}
	\item[-] For the first formula, $\wnext\alwaysF\left( \metricI{\previous} \top \to \left(\previous \Lab{\mu} \leftrightarrow \Lab{\varphi}\right) \right)$, it amounts to prove that 
	$\trivalp{k}{\Lab{\varphi}} = \trivalp{k}{\previous \Lab{\mu}}$, for all $\rangeco{k}{1}{\lambda}$ satisfying $\tau(k)-\tau(k-1) \in I$ (being trivially true when $\lambda=1$). 
	We reason as follows:
	\begin{align*}
		\trivalp{k}{\previous \Lab{\mu}} &= \trivalp{k-1}{\Lab\mu} &\\
		&= \trival{k-1}{\mu}\\
		&= \trival{k-1}{\metricI{\next}\varphi}\\
		&= \trival{k}{\varphi} & \hbox{since } \tau(k)-\tau(k-1) \in I\\
		&= \trivalp{k}{\Lab{\varphi} }
	\end{align*}
	\item[-] For the second formula, $\wnext\alwaysF\left( \neg \metricI{\previous} \top \to \previous \Lab{\mu} \right)$, it amounts to prove that 
	$\trivalp{k}{\previous \Lab{\mu}} = 2$, for all $\rangeco{k}{1}{\lambda}$ satisfying $\tau(k)-\tau(k-1) \not \in I$ (being trivially true when $\lambda=1$).
	We reason as follows:
	\begin{equation*}
		\trivalp{k}{\previous \Lab{\mu}} = \trival{k-1}{\mu} 
		= \trival{k-1}{\metricI{\wnext}\varphi}
		= 2, \hspace{10pt} \hbox{since } \tau(k)-\tau(k-1) \not \in I.
	\end{equation*}
	\item[-] For the third formula, $\alwaysF (\finally \to  \Lab{\mu})$ it amounts to show that $\trivalp{\lambda-1}{\Lab{\mu}}=2$, which follows by construction of ${\bm m}'$, since $\trivalp{\lambda-1}{\Lab{\mu}}=\trival{\lambda-1}{\mu}=\trival{\lambda -1}{\metricI{\wnext}\varphi}=2$.
\end{itemize}
\item For $\mu = \varphi \metric{\until}{0}{n} \psi$, we have one formula in $\eta(\mu)$. In this case, we have to prove that 
		\begin{equation*}
				\trivalp{k}{\Lab{\mu}}=\trivalp{k}{\Lab{\psi} \vee \left( \Lab{\varphi} \wedge \bigvee\limits_{x=1}^n \Lab{\metric{\next}{x}{x} (\varphi\metric{\until}{0}{(n-x)} \psi)}\right)} \\
		\end{equation*}	
		\noindent for all $\rangeco{k}{0}{\lambda}$ and that is proved next:
		\begin{eqnarray*}
			&&\trivalp{k}{\Lab{\mu}} \\
			&=&\trival{k}{\mu}\\
			&=&\trival{k}{\varphi \metric{\until}{0}{n} \psi}\\
			&\stackrel{Prop~\ref{prop:alternative:semantics}}{=}& \max\lbrace \trival{k}{\psi},\min\lbrace  \trival{k}{\varphi},\max\lbrace \trival{k}{\metric{\next}{x}{x} (\varphi\metric{\until}{0}{n-x} \psi)}\mid \rangecc{x}{1}{n}  \rbrace  \rbrace \rbrace\\
			&=& \max\lbrace \trivalp{k}{\Lab\psi},\min\lbrace  \trivalp{k}{\Lab\varphi},\max\lbrace \trivalp{k}{\Lab{\metric{\next}{x}{x} (\varphi\metric{\until}{0}{n-x} \psi)}}\mid \rangecc{x}{1}{n}  \rbrace  \rbrace \rbrace\\
			&=&\trivalp{k}{\Lab{\psi} \vee \left( \Lab{\varphi} \wedge \bigvee\limits_{x=1}^n \Lab{\metric{\next}{x}{x} (\varphi\metric{\until}{0}{(n-x)} \psi)}\right)}
		\end{eqnarray*}
		\item For $\mu = \varphi \metric{\release}{0}{n} \psi$, we have only one formula in $\eta(\mu)$. In this case, we have to prove that 
		\begin{equation*}
			\trivalp{k}{\Lab{\mu}}=\trivalp{k}{\Lab{\psi} \wedge \left( \Lab{\varphi} \vee \bigwedge\limits_{x=1}^n \Lab{\metric{\wnext}{x}{x} (\varphi\metric{\release}{0}{(n-x)} \psi)}\right)} \\
		\end{equation*}	
		\noindent for all $\rangeco{k}{0}{\lambda}$ and that is proved next:
		\begin{eqnarray*}
			&&\trivalp{k}{\Lab{\mu}} \\
			&=&\trival{k}{\mu}\\
			&=&\trival{k}{\varphi \metric{\release}{0}{n} \psi}\\
			&\stackrel{Prop~\ref{prop:alternative:semantics}}{=}& \min\lbrace \trival{k}{\psi},\max\lbrace  \trival{k}{\varphi},\min\lbrace \trival{k}{\metric{\wnext}{x}{x} (\varphi\metric{\release}{0}{n-x} \psi)}\mid \rangecc{x}{1}{n}  \rbrace  \rbrace \rbrace\\
			&=& \min\lbrace \trivalp{k}{\Lab\psi},\max\lbrace  \trivalp{k}{\Lab\varphi},\min\lbrace \trivalp{k}{\Lab{\metric{\wnext}{x}{x} (\varphi\metric{\release}{0}{n-x} \psi)}}\mid \rangecc{x}{1}{n}  \rbrace  \rbrace \rbrace\\
			&=&\trivalp{k}{\Lab{\psi} \wedge \left( \Lab{\varphi} \vee \bigwedge\limits_{x=1}^n \Lab{\metric{\wnext}{x}{x} (\varphi\metric{\release}{0}{(n-x)} \psi)}\right)}.
		\end{eqnarray*}
	\end{enumerate}
\end{proofof}

 \begin{proofof}{Lemma~\ref{lem:nf2}}
	By induction on $\prec_{\dist}$
	\begin{enumerate}
		\item If $\mu$ is $\bot$ or an atom $p$, this is trivial because $\Lab\mu=\mu$ by definition.
		
		\item If $\mu=\varphi \otimes \psi$ for any propositional connective $\otimes \in \{\vee,\wedge,\to\}$ then:
		\[\begin{array}{rll}
			\trival{k}{\Lab\mu} &= \trival{k}{\Lab\varphi \otimes \Lab\psi}
			& \text{using the equality in }\eta(\mu) \\
			&= f^\otimes(\trival{k}{\Lab\varphi},\trival{k}{\Lab\psi}) \\
			&= f^\otimes(\trival{k}{\varphi},\trival{k}{\psi})
			& \text{By induction on } \varphi \prec_\dist \varphi \otimes \psi \hbox{ and } \psi \prec_\dist \varphi \otimes \psi \\
			&= \trival{k}{\varphi \otimes \psi} \\
			&= \trival{k}{\mu}
		\end{array}\]
		
		\item The cases for the non-metric temporal connectives are done as in~\cite{agcadipescscvi20a}

		\item If $\mu=\metricI{\previous} \varphi$ we divide into three cases cases:
		\begin{itemize}
			\item[-] If $k=0$, $\trival{0}{\metricI{\previous} \varphi}=0$.
			Since $\bm{m}$ is a model of $\sigma(\Gamma)$ and $\neg \Lab{\mu} \in \eta(\mu)$ we conclude that 
			$\trival{0}{\neg \Lab{\mu}}=2$. Therefore, $\trival{0}{\Lab{\mu}}=0$.
			
			\item[-]  If $k>0$ and $\tau(k)-\tau(k-1)\in I$ we use the formula $\wnext\alwaysF\left(\Lab{\mu} \leftrightarrow  \metricI{\previous}\Lab{\varphi}\right) \in \eta(\mu)$ to get
			
			\begin{eqnarray*}
			\trival{k}{\Lab{\mu}}&=&\trival{k}{\metricI{\previous} \Lab{\varphi}}\\
								 &=&\trival{k-1}{\Lab{\varphi}} \hspace{10pt}\hbox{ since } \tau(k)-\tau(k-1)\in I \hbox{ and } k>0\\
								 &=&\trival{k-1}{\varphi} \hspace{10pt}\hbox{ since } \varphi \prec_\dist \metricI{\previous} \varphi \hbox{ and IH} 
			\end{eqnarray*}
			
			\item[-] If $k>0$ and $\tau(k)-\tau(k-1)\not \in I$ we clearly conclude that that $\trival{k}{\metricI{\previous}\varphi} = 0$. 
			We use now the formula $\wnext\alwaysF\left(\Lab{\mu} \leftrightarrow  \metricI{\previous}\Lab{\varphi}\right) \in \eta(\mu)$ to get $\trival{k}{\Lab{\mu}}=\trival{k}{\metricI{\previous} \Lab{\varphi}} =0$ (since $\tau(k)-\tau(k-1)\not \in I$).
		\end{itemize}		
		\item If $\mu=\metricI{\wprevious} \varphi$, we consider the following cases:
		\begin{itemize}
			\item[-] If $k=0$, $\trival{0}{\metricI{\wprevious} \varphi} =2$ because of the semantics. In this case we use the formula $\Lab{\mu} \in \eta(\mu)$ to readily conclude that 
				$\trival{0}{\Lab{\mu}} = 2$.			
			\item[-]  If $k>0$ and $\tau(k)-\tau(k-1)\in I$, we use $\wnext\alwaysF\left(\Lab{\mu} \leftrightarrow  \metricI{\wprevious}\Lab{\varphi}\right) \in \eta(\mu)$ to get: 
			\begin{eqnarray*}
				\trival{k}{\Lab{\mu}}&=&\trival{k}{\metricI{\wprevious} \Lab{\varphi}}\\
					&=&\trival{k-1}{\Lab{\varphi}}\hspace{10pt}\hbox{ since } \tau(k)-\tau(k-1)\in I \hbox{ and } k > 0.\\
					&=&\trival{k-1}{\varphi} \hspace{10pt}\hbox{ since } \varphi \prec_\dist \metricI{\wprevious} \varphi  \hbox{ and IH}. \\
			\end{eqnarray*}
			\item[-] If $k>0$ and $\tau(k)-\tau(k-1)\not \in I$ we clearly conclude that that $\trival{k}{\metricI{\wprevious}\varphi} = 2$. 
			We use (again) the formula $\wnext\alwaysF\left(\Lab{\mu} \leftrightarrow  \metricI{\wprevious}\Lab{\varphi}\right) \in \eta(\mu)$ to get $\trival{k}{\Lab{\mu}}=\trival{k}{\metricI{\wprevious} \Lab{\varphi}} =2$ (since $\tau(k)-\tau(k-1)\not \in I$).
		\end{itemize}
		\item If $\mu=\varphi \metric{\since}{0}{n} \psi$, we will use the formula 
		\begin{equation}
			\alwaysF\left(\Lab{\mu}\leftrightarrow \left(\Lab{\psi} \vee   \left(\Lab{\varphi} \wedge \bigvee\limits_{x=1}^n \Lab{\metric{\previous}{x}{x} (\varphi\metric{\since}{0}{(n-x)} \psi)}\right)\right)\right) \in \eta(\mu).\label{eq:formula:since}
		\end{equation}
		We also remark that $\metric{\previous}{x}{x} ( \varphi\metric{\since}{0}{(n-x)} \psi) \in \cl{\mu}$, so $\eta(\metric{\previous}{x}{x} (\varphi\metric{\since}{0}{(n-x)}\psi)) \subseteq \eta(\mu)$ for all $\rangecc{x}{1}{n}$. We distinguish two cases:			
		\begin{itemize}
			\item[-] If $k=0$, we use the formulas $\neg \Lab{ \metric{\previous}{x}{x} (\varphi\metric{\since}{0}{(n-x)} \psi} \in \eta(\mu)$,
			for all $\rangecc{x}{1}{n}$, to conclude that $\trival{k}{\bigvee\limits_{x=1}^n \Lab{\metric{\previous}{x}{x} (\varphi\metric{\since}{0}{(n-x)} \psi)}} = 0$.
			Therereore, $\trival{k}{\eqref{eq:formula:since}} = \trival{k}{\Lab\psi} \stackrel{(IH)}{=} \trival{k}{\psi}=\trival{k}{\varphi \metric{\since}{0}{n}\psi}$. 
			\item[-] If $\rangeco{k}{1}{\lambda}$, we use~\eqref{eq:formula:release} to conclude that 			
			\begin{eqnarray*}
				&&\trival{k}{\Lab{\mu}}\nonumber\\ 
				&=& \max\lbrace \trival{k}{\Lab{\psi}},\min\lbrace  \trival{k}{\Lab{\varphi}},\max\lbrace \trival{k}{\Lab{\metric{\previous}{x}{x} (\varphi\metric{\since}{0}{n-x} \psi)} }\mid \rangecc{x}{1}{n}  \rbrace  \rbrace \rbrace \nonumber\\
				&=& \max\lbrace \trival{k}{\psi},\min\lbrace  \trival{k}{\varphi},\max\lbrace \trival{k}{\Lab{\metric{\previous}{x}{x} (\varphi\metric{\since}{0}{n-x} \psi)} }\mid \rangecc{x}{1}{n}  \rbrace  \rbrace \rbrace, 
			\end{eqnarray*}
			\noindent by induction\footnote{Note that $\varphi \prec_\dist \varphi \metric{\since}{0}{n} \psi$ and  $\psi \prec_\dist \varphi \metric{\since}{0}{n} \psi$.}. 	
			We will prove now that $\trival{k}{\Lab{\metric{\previous}{x}{x} (\varphi\metric{\since}{0}{n-x} \psi)} }=\trival{k}{\metric{\previous}{x}{x} (\varphi\metric{\since}{0}{n-x} \psi) }$, for all $\rangecc{x}{1}{n}$. This can be directly done by using the formula 
			\begin{equation*}
				\wnext\alwaysF\left(\Lab{\metric{\previous}{x}{x} (\varphi\metric{\since}{0}{n-x} \psi)} \leftrightarrow  \metric{\previous}{x}{x}\Lab{(\varphi\metric{\since}{0}{n-x} \psi)}\right) \in \eta(\mu). 
			\end{equation*}

			\begin{enumerate}
				\item If $\tau(k)-\tau(k-1) = x$, we reason as follows 
				
				\begin{eqnarray*}
					\trival{k}{\Lab{\metric{\previous}{x}{x} (\varphi\metric{\since}{0}{n-x} \psi)}} &=& \trival{k}{\metric{\previous}{x}{x} 	\Lab{\varphi\metric{\since}{0}{n-x} \psi}}\\
																									 &=& \trival{k-1}{\Lab{\varphi\metric{\since}{0}{n-x} \psi}}\\
																									 & \stackrel{(IH)}{=}& \trival{k-1}{\varphi\metric{\since}{0}{n-x} \psi}\\
																									 &=&\trival{k}{\metric{\previous}{x}{x}\left(\varphi\metric{\since}{0}{n-x} \psi\right)}.
				\end{eqnarray*} 
				\item If $\tau(k)-\tau(k-1) \not = x$, the entailment relation tells us that $\trival{k}{\metric{\previous}{x}{x}\left(\varphi\metric{\since}{0}{n-x} \psi\right)} = 0$. The same result for $\Lab{\metric{\previous}{x}{x} (\varphi\metric{\since}{0}{n-x} \psi)}$ can be obtained:
				\begin{equation*}
					\trival{k}{\Lab{\metric{\previous}{x}{x} (\varphi\metric{\since}{0}{n-x} \psi)}} = \trival{k}{\metric{\previous}{x}{x} 	\Lab{\varphi\metric{\since}{0}{n-x} \psi}} = 0.
				\end{equation*}
			\end{enumerate}
			
			Since $\Lab{\metric{\previous}{x}{x} (\varphi\metric{\since}{0}{n-x} \psi)}$ was chosen arbitrarily, it follows that 		
			\begin{eqnarray*}
				&&\max\lbrace \trival{k}{\Lab{\metric{\previous}{x}{x} (\varphi\metric{\since}{0}{n-x} \psi)} }\mid \rangecc{x}{1}{n}  \rbrace \\
				&=& \max\lbrace \trival{k}{\metric{\previous}{x}{x} (\varphi\metric{\since}{0}{n-x} \psi)}\mid \rangecc{x}{1}{n}  \rbrace.
			\end{eqnarray*}
			\noindent We can use now this fact to conclude that, 		
			\begin{eqnarray*}
				&&\trival{k}{\Lab{\mu}}\nonumber\\ 
				&=& \max\lbrace \trival{k}{\Lab{\psi}},\min\lbrace  \trival{k}{\Lab{\varphi}},\max\lbrace \trival{k}{\Lab{\metric{\previous}{x}{x} (\varphi\metric{\since}{0}{n-x} \psi)} }\mid \rangecc{x}{1}{n}  \rbrace  \rbrace \rbrace \nonumber\\
				&=& \max\lbrace \trival{k}{\psi},\min\lbrace  \trival{k}{\varphi},\max\lbrace \trival{k}{\Lab{\metric{\previous}{x}{x} (\varphi\metric{\since}{0}{n-x} \psi)} }\mid \rangecc{x}{1}{n}  \rbrace  \rbrace \rbrace\\
				&=& \max\lbrace \trival{k}{\psi},\min\lbrace  \trival{k}{\varphi},\max\lbrace \trival{k}{\metric{\previous}{x}{x} (\varphi\metric{\since}{0}{n-x} \psi) }\mid \rangecc{x}{1}{n}  \rbrace  \rbrace \rbrace\\
				&\stackrel{Prop.~\ref{prop:alternative:semantics}}{=}&	\trival{k}{\mu}.			
			\end{eqnarray*}
		\end{itemize}

	\item If $\mu=\varphi \metric{\trigger}{0}{n} \psi$, we will use the formula 
	\begin{equation}
		\alwaysF\left(\Lab{\mu}\leftrightarrow \left(\Lab{\psi} \wedge   \left(\Lab{\varphi} \vee \bigwedge\limits_{x=1}^n \Lab{\metric{\wprevious}{x}{x} (\varphi\metric{\trigger}{0}{(n-x)} \psi)}\right)\right)\right) \in \eta(\mu).\label{eq:formula:trigger}
	\end{equation}
	We also remark that $\metric{\wprevious}{x}{x} ( \varphi\metric{\trigger}{0}{(n-x)} \psi) \in \cl{\mu}$, so $\eta(\metric{\wprevious}{x}{x} (\varphi\metric{\trigger}{0}{(n-x)}\psi)) \subseteq \eta(\mu)$ for all $\rangecc{x}{1}{n}$. We distinguish two cases:			
	\begin{itemize}
		\item[-] If $k=0$, we use the formulas $\Lab{ \metric{\wprevious}{x}{x} (\varphi\metric{\trigger}{0}{(n-x)} \psi} \in \eta(\mu)$,
		for all $\rangecc{x}{1}{n}$, to conclude that $\trival{k}{\bigwedge\limits_{x=1}^n \Lab{\metric{\wprevious}{x}{x} (\varphi\metric{\trigger}{0}{(n-x)} \psi)}} = 2$.
		Therereore, $\trival{k}{\eqref{eq:formula:trigger}} = \trival{k}{\Lab\psi} \stackrel{(IH)}{=} \trival{k}{\psi}=\trival{k}{\varphi \metric{\trigger}{0}{n}\psi}$. 
		\item[-] If $\rangeco{k}{1}{\lambda}$, we use~\eqref{eq:formula:release} to conclude that 			
		\begin{eqnarray*}
			&&\trival{k}{\Lab{\mu}}\nonumber\\ 
			&=& \min\lbrace \trival{k}{\Lab{\psi}},\max\lbrace  \trival{k}{\Lab{\varphi}},\min\lbrace \trival{k}{\Lab{\metric{\wprevious}{x}{x} (\varphi\metric{\trigger}{0}{n-x} \psi)} }\mid \rangecc{x}{1}{n}  \rbrace  \rbrace \rbrace \nonumber\\
			&=& \min\lbrace \trival{k}{\psi},\max\lbrace  \trival{k}{\varphi},\min\lbrace \trival{k}{\Lab{\metric{\wprevious}{x}{x} (\varphi\metric{\trigger}{0}{n-x} \psi)} }\mid \rangecc{x}{1}{n}  \rbrace  \rbrace \rbrace, 
		\end{eqnarray*}
		\noindent by induction\footnote{Note that $\varphi \prec_\dist \varphi \metric{\trigger}{0}{n} \psi$ and  $\psi \prec_\dist \varphi \metric{\trigger}{0}{n} \psi$.}. 	
		We will prove now that $\trival{k}{\Lab{\metric{\wprevious}{x}{x} (\varphi\metric{\trigger}{0}{n-x} \psi)} }=\trival{k}{\metric{\wprevious}{x}{x} (\varphi\metric{\trigger}{0}{n-x} \psi) }$, for all $\rangecc{x}{1}{n}$. This can be directly done by using the formula 
		\begin{equation*}
			\wnext\alwaysF\left(\Lab{\metric{\wprevious}{x}{x} (\varphi\metric{\trigger}{0}{n-x} \psi)} \leftrightarrow  \metric{\wprevious}{x}{x}\Lab{(\varphi\metric{\trigger}{0}{n-x} \psi)}\right) \in \eta(\mu). 
		\end{equation*}			
				
		\begin{enumerate}
			\item If $\tau(k)-\tau(k-1) = x$, we reason as follows 
			
			\begin{eqnarray*}
				\trival{k}{\Lab{\metric{\wprevious}{x}{x} (\varphi\metric{\trigger}{0}{n-x} \psi)}} &=& \trival{k}{\metric{\wprevious}{x}{x} 	\Lab{\varphi\metric{\trigger}{0}{n-x} \psi}}\\
				&=& \trival{k-1}{\Lab{\varphi\metric{\trigger}{0}{n-x} \psi}}\\
				& \stackrel{(IH)}{=}& \trival{k-1}{\varphi\metric{\trigger}{0}{n-x} \psi}\\
				&=&\trival{k}{\metric{\wprevious}{x}{x}\left(\varphi\metric{\trigger}{0}{n-x} \psi\right)}.
			\end{eqnarray*} 
			\item If $\tau(k)-\tau(k-1) \not = x$, the entailment relation tells us that $\trival{k}{\metric{\wprevious}{x}{x}\left(\varphi\metric{\trigger}{0}{n-x} \psi\right)} = 2$. The same result for $\Lab{\metric{\wprevious}{x}{x} (\varphi\metric{\trigger}{0}{n-x} \psi)}$ can be obtained:
			\begin{equation*}
				\trival{k}{\Lab{\metric{\wprevious}{x}{x} (\varphi\metric{\trigger}{0}{n-x} \psi)}} = \trival{k}{\metric{\wprevious}{x}{x} 	\Lab{\varphi\metric{\trigger}{0}{n-x} \psi}} = 2.
			\end{equation*}
		\end{enumerate}		
		Since $\Lab{\metric{\wprevious}{x}{x} (\varphi\metric{\trigger}{0}{n-x} \psi)}$ was chosen arbitrarily, it follows that 		
		\begin{eqnarray*}
			&&\min\lbrace \trival{k}{\Lab{\metric{\wprevious}{x}{x} (\varphi\metric{\trigger}{0}{n-x} \psi)} }\mid \rangecc{x}{1}{n}  \rbrace \\
			&=& \min\lbrace \trival{k}{\metric{\wprevious}{x}{x} (\varphi\metric{\trigger}{0}{n-x} \psi)}\mid \rangecc{x}{1}{n}  \rbrace.
		\end{eqnarray*}
		\noindent We can use now this fact to conclude that, 		
		\begin{eqnarray*}
			&&\trival{k}{\Lab{\mu}}\nonumber\\ 
			&=& \min\lbrace \trival{k}{\Lab{\psi}},\max\lbrace  \trival{k}{\Lab{\varphi}},\min\lbrace \trival{k}{\Lab{\metric{\wprevious}{x}{x} (\varphi\metric{\trigger}{0}{n-x} \psi)} }\mid \rangecc{x}{1}{n}  \rbrace  \rbrace \rbrace \nonumber\\
			&=& \min\lbrace \trival{k}{\psi},\max\lbrace  \trival{k}{\varphi},\min\lbrace \trival{k}{\Lab{\metric{\wprevious}{x}{x} (\varphi\metric{\trigger}{0}{n-x} \psi)} }\mid \rangecc{x}{1}{n}  \rbrace  \rbrace \rbrace\\
			&=& \min\lbrace \trival{k}{\psi},\max\lbrace  \trival{k}{\varphi},\min\lbrace \trival{k}{\metric{\wprevious}{x}{x} (\varphi\metric{\trigger}{0}{n-x} \psi) }\mid \rangecc{x}{1}{n}  \rbrace  \rbrace \rbrace\\
			&\stackrel{Prop.~\ref{prop:alternative:semantics}}{=}&	\trival{k}{\mu}.			
		\end{eqnarray*}
	\end{itemize}
	 
		\item If $\mu=\metricI{\next} \varphi$ we divide the proof into three cases:
		\begin{itemize}
			\item[-] If $k=\lambda -1$, $\trival{k}{\metricI{\next}\varphi} = 0$.
			We use the expression $\alwaysF\left( \finally \to \neg \Lab{\mu}\right) \in \eta(\mu)$ to conclude that $\trival{k}{\neg \Lab\mu} = 2$, which means that $\trival{k}{\Lab\mu} = 0$.

			\item[-]  If $\rangeco{k}{0}{\lambda-1}$ and $\tau(k+1)-\tau(k) \in I$ we use the formula 
				\begin{equation*}
						\wnext\alwaysF\left( \metricI{\previous} \top \to \left(\previous \Lab{\mu} \leftrightarrow \Lab{\varphi}\right) \right)\in \eta(\mu)
				\end{equation*}
			\noindent to get that
			$\trival{k+1}{\previous \Lab{\mu}} = \trival{k+1}{\Lab{\varphi}}$. Since $\varphi \prec_\dist \metricI{\next}\varphi$ then we can apply induction 
			to get $\trival{k+1}{\previous \Lab{\mu}} = \trival{k+1}{\varphi}$. By the semantics we get that $\trival{k}{\Lab{\mu}} = \trival{k+1}{\varphi}$.
			Since $\tau(k+1)-\tau(k) \in I$ it follows $\trival{k}{\Lab{\mu}} = \trival{k}{\metricI{\next}\varphi}$.
			\item[-]  If $\rangeco{k}{0}{\lambda-1}$ and $\tau(k+1)-\tau(k) \not \in I$, $\trival{k}{\metricI{\next} \varphi} = 0$.
			We consider the formula  $\wnext\alwaysF\left( \neg \metricI{\previous} \top \to \previous \neg \Lab{\mu} \right) \in \eta(\mu)$.
			Since $\tau(k+1)-\tau(k) \not \in I$, $\trival{k+1}{\metricI{\previous} \top}=0$ so $\trival{k+1}{\neg \metricI{\previous} \top}=2$.
			By Modus Ponens, $\trival{k+1}{\previous \neg \Lab{\mu}} = 2$. By the semantics, $\trival{k}{\neg \Lab{\mu}} = 2$ so $\trival{k}{\Lab{\mu}} = 0$ as required.
		\end{itemize}		
	\item If $\mu=\metricI{\wnext} \varphi$ we divide the proof into three cases:

	\begin{itemize}
	\item[-] If $k=\lambda -1$, $\trival{k}{\wnext\varphi} = 2$ because of the semantics. We use now the formula $\alwaysF\left( \finally \to  \Lab{\mu}\right)  \in \eta(\mu)$ so conclude $\trival{k}{\Lab\mu}=2$. 
	\item[-]  If $\rangeco{k}{0}{\lambda-1}$ and $\tau(k+1)-\tau(k) \in I$ we use the first formula in $\eta(\mu)$ to get that
	$\trival{k+1}{\previous \Lab{\mu}} = \trival{k+1}{\Lab{\varphi}}$. Since $\varphi \prec_\dist \metricI{\wnext}\varphi$ then we can apply induction 
	to get $\trival{k+1}{\previous \Lab{\mu}} = \trival{k+1}{\varphi}$. By the semantics we get that $\trival{k}{\Lab{\mu}} = \trival{k+1}{\varphi}$.
	Since $\tau(k+1)-\tau(k) \in I$ it follows $\trival{k}{\Lab{\mu}} = \trival{k}{\metricI{\wnext}\varphi}$.
	\item[-]  If $\rangeco{k}{0}{\lambda-1}$ and $\tau(k+1)-\tau(k) \not \in I$, $\trival{k}{\wnext \varphi} = 2$. 
	We use not the formula $\wnext\alwaysF\left( \neg \metricI{\previous} \top \to  \previous \Lab{\mu} \right)\in \eta(\mu)$.
	Since $\tau(k+1)-\tau(k) \not \in I$, $\trival{k+1}{\neg \previous \top} = 2$ so $\trival{k+1}{\previous \Lab \mu}=2$ because of Modus Ponens.
	By the semantics, $\trival{k}{\Lab\mu} = 2$ as required.	
\end{itemize}
	\item If $\mu=\varphi \metric{\until}{0}{n} \psi$. To prove this case we will use the formula 
		\begin{equation}
				\alwaysF\left(\Lab{\mu}\leftrightarrow \Lab{\psi} \vee \left( \Lab{\varphi} \wedge \bigvee\limits_{x=1}^n \Lab{\metric{\next}{x}{x} (\varphi\metric{\until}{0}{(n-x)} \psi)}\right)\right) \in \eta(\mu).\label{eq:formula:until}
		\end{equation}
		We also remark that $\metric{\next}{x}{x} ( \varphi\metric{\until}{0}{(n-x)} \psi) \in \cl{\mu}$, so $\eta(\metric{\next}{x}{x} (\varphi\metric{\until}{0}{(n-x)}\psi)) \subseteq \eta(\mu)$ for all $\rangecc{x}{1}{n}$. We distinguish two cases:	

		\begin{itemize}
			\item[-] If $k=\lambda -1$, we use the formulas $\alwaysF\left( \finally \to \neg \Lab{ \metric{\next}{x}{x} (\varphi\metric{\until}{0}{(n-x)} \psi} \right) \in \eta(\mu)$,
			for all $\rangecc{x}{1}{n}$, to conclude that $\trival{k}{\bigvee\limits_{x=1}^n \Lab{\metric{\next}{x}{x} (\varphi\metric{\until}{0}{(n-x)} \psi)}} = 0$.
			Therereore, $\trival{k}{\eqref{eq:formula:until}} = \trival{k}{\Lab\psi} \stackrel{(IH)}{=} \trival{k}{\psi}=\trival{k}{\varphi \metric{\until}{0}{n}\psi}$. 
			\item[-] If $\rangeco{k}{0}{\lambda-1}$, we use~\eqref{eq:formula:until} to conclude that 			
			\begin{eqnarray*}
				&&\trival{k}{\Lab{\mu}}\nonumber\\ 
				&=& \max\lbrace \trival{k}{\Lab{\psi}},\min\lbrace  \trival{k}{\Lab{\varphi}},\max\lbrace \trival{k}{\Lab{\metric{\next}{x}{x} (\varphi\metric{\until}{0}{n-x} \psi)} }\mid \rangecc{x}{1}{n}  \rbrace  \rbrace \rbrace \nonumber\\
				&=& \max\lbrace \trival{k}{\psi},\min\lbrace  \trival{k}{\varphi},\max\lbrace \trival{k}{\Lab{\metric{\next}{x}{x} (\varphi\metric{\until}{0}{n-x} \psi)} }\mid \rangecc{x}{1}{n}  \rbrace  \rbrace \rbrace, 
			\end{eqnarray*}
			\noindent by induction\footnote{Note that $\varphi \prec_\dist \varphi \metric{\until}{0}{n} \psi$ and  $\psi \prec_\dist \varphi \metric{\until}{0}{n} \psi$.}. 	
			We will prove now that $\trival{k}{\Lab{\metric{\next}{x}{x} (\varphi\metric{\until}{0}{n-x} \psi)} }=\trival{k}{\metric{\next}{x}{x} (\varphi\metric{\until}{0}{n-x} \psi) }$, for all $\rangecc{x}{1}{n}$. 	
			\begin{enumerate}
				\item If $\tau(k+1)-\tau(k) = x$ then  we use the formula 
				\begin{equation*}
				\wnext\alwaysF\left( \metric{\previous}{x}{x} \top \to \left(\previous \Lab{\metric{\next}{x}{x} (\varphi\metric{\until}{0}{n-x} \psi)} \leftrightarrow \Lab{(\varphi\metric{\until}{0}{n-x} \psi)}\right) \right) \in \eta(\mu)
				\end{equation*}
			\noindent to conclude that $\trival{k+1}{\previous \Lab{\metric{\next}{x}{x} (\varphi\metric{\until}{0}{n-x} \psi)}} = \trival{k+1}{\Lab{(\varphi\metric{\until}{0}{n-x} \psi)}}$.
			By induction, $\trival{k+1}{\previous \Lab{\metric{\next}{x}{x} (\varphi\metric{\until}{0}{n-x} \psi)}} = \trival{k+1}{\varphi\metric{\until}{0}{n-x} \psi}$\footnote{$\varphi\metric{\until}{0}{n-x} \psi \prec_\dist \mu$}.
			Since $\tau(k+1)-\tau(k) = x$, $\trival{k}{\Lab{\metric{\next}{x}{x} (\varphi\metric{\until}{0}{n-x} \psi)}} = \trival{k}{\metric{\next}{x}{x}\left(\varphi\metric{\until}{0}{n-x} \psi\right)}$.
			\item If $\tau(k+1)-\tau(k) \not = x$, $\trival{k}{\metric{\next}{x}{x}\left(\varphi\metric{\until}{0}{n-x} \psi\right)} = 0$. We use now the formula $\wnext\alwaysF\left( \neg \metric{\previous}{n}{n} \top \to \previous \neg \Lab{\metric{\next}{x}{x} (\varphi\metric{\until}{0}{n-x} \psi)} \right) \in \eta(\mu)$ to conclude that 
			$2 = \trival{k+1}{\previous \neg \Lab{\metric{\next}{x}{x} (\varphi\metric{\until}{0}{n-x} \psi)}} = \trival{k}{\neg \Lab{\metric{\next}{x}{x} (\varphi\metric{\until}{0}{n-x} \psi)}}$.
			By the semantics, $\trival{k}{\Lab{\metric{\next}{x}{x} (\varphi\metric{\until}{0}{n-x} \psi)}} = 0$. 
			\end{enumerate}
			
			Since $\Lab{\metric{\next}{x}{x} (\varphi\metric{\until}{0}{n-x} \psi)}$ was chosen arbitrarily, it follows that 		
				\begin{eqnarray*}
					&&\max\lbrace \trival{k}{\Lab{\metric{\next}{x}{x} (\varphi\metric{\until}{0}{n-x} \psi)} }\mid \rangecc{x}{1}{n}  \rbrace \\
					&=& \max\lbrace \trival{k}{\metric{\next}{x}{x} (\varphi\metric{\until}{0}{n-x} \psi)}\mid \rangecc{x}{1}{n}  \rbrace.
				\end{eqnarray*}
			\noindent We can use now this fact to conclude that, 		
			\begin{eqnarray*}
				&&\trival{k}{\Lab{\mu}}\nonumber\\ 
				&=& \max\lbrace \trival{k}{\Lab{\psi}},\min\lbrace  \trival{k}{\Lab{\varphi}},\max\lbrace \trival{k}{\Lab{\metric{\next}{x}{x} (\varphi\metric{\until}{0}{n-x} \psi)} }\mid \rangecc{x}{1}{n}  \rbrace  \rbrace \rbrace \nonumber\\
				&=& \max\lbrace \trival{k}{\psi},\min\lbrace  \trival{k}{\varphi},\max\lbrace \trival{k}{\Lab{\metric{\next}{x}{x} (\varphi\metric{\until}{0}{n-x} \psi)} }\mid \rangecc{x}{1}{n}  \rbrace  \rbrace \rbrace\\
				&=& \max\lbrace \trival{k}{\psi},\min\lbrace  \trival{k}{\varphi},\max\lbrace \trival{k}{\metric{\next}{x}{x} (\varphi\metric{\until}{0}{n-x} \psi) }\mid \rangecc{x}{1}{n}  \rbrace  \rbrace \rbrace\\
				&\stackrel{Prop.~\ref{prop:alternative:semantics}}{=}&	\trival{k}{\mu}.			
			\end{eqnarray*}
		\end{itemize}
		\item The case for $\mu=\varphi \metric{\release}{0}{n} \psi$ is similar to the case of $\mu=\psi \metric{\until}{0}{n} \varphi$. To prove this case we will use the formula 
		\begin{equation}
			 \wnext\alwaysF\left(\Lab{\mu}\leftrightarrow \Lab{\psi} \wedge \left( \Lab{\varphi} \vee \bigwedge\limits_{x=1}^n \Lab{\metric{\wnext}{x}{x} (\varphi\metric{\release}{0}{(n-x)} \psi)}\right)\right)\in \eta(\mu). \label{eq:formula:release}
		\end{equation}		
		We also remark that $\metric{\wnext}{x}{x} ( \varphi\metric{\release}{0}{(n-x)} \psi) \in \cl{\mu}$, so $\eta(\metric{\wnext}{x}{x} (\varphi\metric{\release}{0}{(n-x)}\psi)) \subseteq \eta(\mu)$ for all $\rangecc{x}{1}{n}$. We distinguish two cases:	
		
		\begin{itemize}
			\item[-] If $k=\lambda -1$, we use the formulas $\alwaysF\left( \finally \to  \Lab{ \metric{\wnext}{x}{x} (\varphi\metric{\release}{0}{(n-x)} \psi} \right) \in \eta(\mu)$,
			for all $\rangecc{x}{1}{n}$, to conclude that $\trival{k}{\bigwedge\limits_{x=1}^n \Lab{\metric{\wnext}{x}{x} (\varphi\metric{\release}{0}{(n-x)} \psi)}} = 2$.
			Therereore, $\trival{k}{\eqref{eq:formula:release}} = \trival{k}{\Lab\psi} \stackrel{(IH)}{=} \trival{k}{\psi}=\trival{k}{\varphi \metric{\release}{0}{n}\psi}$. 
			\item[-] If $\rangeco{k}{0}{\lambda-1}$, we use~\eqref{eq:formula:release} to conclude that 			
			\begin{eqnarray*}
				&&\trival{k}{\Lab{\mu}}\nonumber\\ 
				&=& \min\lbrace \trival{k}{\Lab{\psi}},\max\lbrace  \trival{k}{\Lab{\varphi}},\min\lbrace \trival{k}{\Lab{\metric{\wnext}{x}{x} (\varphi\metric{\release}{0}{n-x} \psi)} }\mid \rangecc{x}{1}{n}  \rbrace  \rbrace \rbrace \nonumber\\
				&=& \min\lbrace \trival{k}{\psi},\max\lbrace  \trival{k}{\varphi},\min\lbrace \trival{k}{\Lab{\metric{\wnext}{x}{x} (\varphi\metric{\release}{0}{n-x} \psi)} }\mid \rangecc{x}{1}{n}  \rbrace  \rbrace \rbrace, 
			\end{eqnarray*}
			\noindent by induction\footnote{Note that $\varphi \prec_\dist \varphi \metric{\release}{0}{n} \psi$ and  $\psi \prec_\dist \varphi \metric{\release}{0}{n} \psi$.}. 	
			We will prove now that $\trival{k}{\Lab{\metric{\wnext}{x}{x} (\varphi\metric{\release}{0}{n-x} \psi)} }=\trival{k}{\metric{\wnext}{x}{x} (\varphi\metric{\release}{0}{n-x} \psi) }$, for all $\rangecc{x}{1}{n}$. 	
			\begin{enumerate}
				\item If $\tau(k+1)-\tau(k) = x$ then  we use the formula 
				\begin{equation*}
					\wnext\alwaysF\left( \metric{\previous}{x}{x} \top \to \left(\previous \Lab{\metric{\wnext}{x}{x} (\varphi\metric{\release}{0}{n-x} \psi)} \leftrightarrow \Lab{(\varphi\metric{\release}{0}{n-x} \psi)}\right) \right) \in \eta(\mu)
				\end{equation*}
				\noindent to conclude that $\trival{k+1}{\previous \Lab{\metric{\wnext}{x}{x} (\varphi\metric{\release}{0}{n-x} \psi)}} = \trival{k+1}{\Lab{(\varphi\metric{\release}{0}{n-x} \psi)}}$.
				By induction, $\trival{k+1}{\previous \Lab{\metric{\wnext}{x}{x} (\varphi\metric{\release}{0}{n-x} \psi)}} = \trival{k+1}{\varphi\metric{\release}{0}{n-x} \psi}$\footnote{$\varphi\metric{\release}{0}{n-x} \psi \prec_\dist \mu$}.
				Since $\tau(k+1)-\tau(k) = x$, $\trival{k}{\Lab{\metric{\wnext}{x}{x} (\varphi\metric{\release}{0}{n-x} \psi)}} = \trival{k}{\metric{\wnext}{x}{x}\left(\varphi\metric{\release}{0}{n-x} \psi\right)}$.
				\item If $\tau(k+1)-\tau(k) \not = x$, $\trival{k}{\metric{\wnext}{x}{x}\left(\varphi\metric{\release}{0}{n-x} \psi\right)} = 2$. 
				We use now the formula $\wnext\alwaysF\left( \neg \metric{\previous}{n}{n} \top \to \previous \Lab{\metric{\wnext}{x}{x} (\varphi\metric{\release}{0}{n-x} \psi)} \right) \in \eta(\mu)$ to conclude that 
				$2 = \trival{k+1}{\previous \Lab{\metric{\wnext}{x}{x} (\varphi\metric{\release}{0}{n-x} \psi)}} = \trival{k}{\Lab{\metric{\wnext}{x}{x} (\varphi\metric{\release}{0}{n-x} \psi)}}$.
			\end{enumerate}			
			Since $\Lab{\metric{\wnext}{x}{x} (\varphi\metric{\release}{0}{n-x} \psi)}$ was chosen arbitrarily, it follows that 		
			\begin{eqnarray*}
				&&\min\lbrace \trival{k}{\Lab{\metric{\wnext}{x}{x} (\varphi\metric{\release}{0}{n-x} \psi)} }\mid \rangecc{x}{1}{n}  \rbrace \\
				&=&\min\lbrace \trival{k}{\metric{\wnext}{x}{x} (\varphi\metric{\release}{0}{n-x} \psi)}\mid \rangecc{x}{1}{n}  \rbrace.
			\end{eqnarray*}
			\noindent We can use now this fact to conclude that, 		
			\begin{eqnarray*}
				&&\trival{k}{\Lab{\mu}}\nonumber\\ 
				&=& \min\lbrace \trival{k}{\Lab{\psi}},\max\lbrace  \trival{k}{\Lab{\varphi}},\min\lbrace \trival{k}{\Lab{\metric{\wnext}{x}{x} (\varphi\metric{\release}{0}{n-x} \psi)} }\mid \rangecc{x}{1}{n}  \rbrace  \rbrace \rbrace \nonumber\\
				&=& \min\lbrace \trival{k}{\psi},\max\lbrace  \trival{k}{\varphi},\min\lbrace \trival{k}{\Lab{\metric{\wnext}{x}{x} (\varphi\metric{\release}{0}{n-x} \psi)} }\mid \rangecc{x}{1}{n}  \rbrace  \rbrace \rbrace\\
				&=& \min\lbrace \trival{k}{\psi},\max\lbrace  \trival{k}{\varphi},\min\lbrace \trival{k}{\metric{\wnext}{x}{x} (\varphi\metric{\release}{0}{n-x} \psi) }\mid \rangecc{x}{1}{n}  \rbrace  \rbrace \rbrace\\
				&\stackrel{Prop.~\ref{prop:alternative:semantics}}{=}&	\trival{k}{\mu}.			
			\end{eqnarray*}
		\end{itemize}
\end{enumerate}
\end{proofof}

 \begin{proofof}{Theorem~\ref{them:polynomial}}
$\eta^*(\varphi)$ is computed with respect to $cl(\varphi)$ whose size can be bounded by a polynomial depending on the maximum $n$ use in the intervals occurring in $\varphi$ and $\lvert \varphi \rvert$. 
For each formula  $\mu \in cl(\varphi)$ we generate an associated set of rules: for the case of Boolean connectives and formulas of the form $\metricI{\next}\varphi$, $\metricI{\wnext}\varphi$, $\metricI{\previous}\varphi$ and $\metricI{\wprevious}\varphi$, the number is constant and depends on $\lvert \mu\rvert$.
In the case of formulas of the form $\varphi \metric{\until}{0}{n}\varphi$, $\varphi \metric{\release}{0}{n}\varphi$, $\varphi \metric{\since}{0}{n}\varphi$ and $\varphi \metric{\trigger}{0}{n}\varphi$, the number of generated formulas depends on $n$ and it is bounded to $2n+6$. Therefore, we can bound the number of generated rules by a polynomial bound that depends on two parameters: $n$ and $\lvert \varphi \rvert$.
\end{proofof}
 \section{Translation of Non-Metric Temporal Connectives}
\begin{table}[h!]\centering
\caption{Translation of the propositional connectives.}
\label{tbl:tseitin:prop}
{\tablefont \begin{tabular}{@{\extracolsep{\fill}}ccc}\topline
$\mu$ & $\eta(\mu)$ & $\eta^*(\mu)$ \\\hline
$\varphi \wedge \psi$ & $\alwaysF (\Lab{\mu} \leftrightarrow \Lab{\varphi} \wedge \Lab{\psi})$  & $\begin{array}{l} 
 \wnext \alwaysF( \Lab{\mu} \to   \Lab{\varphi} )\\
\wnext \alwaysF ( \Lab{\mu}  \to   \Lab{\psi} )\\
\wnext \alwaysF( \Lab{\varphi} \wedge  \Lab{\psi}  \to \Lab{\mu})\\
\Lab{\mu}    \to   \Lab{\varphi}\\
\Lab{\mu}   \to    \Lab{\psi} \\
\Lab{\varphi} \wedge \Lab{\psi}  \to   \Lab{\mu} \end{array}$\\\hline
$\varphi \vee \psi$ & $\alwaysF (\Lab{\mu} \leftrightarrow \Lab{\varphi} \vee \Lab{\psi})$ &
$\begin{array}{l}
  \wnext \alwaysF( \Lab{\varphi}  \to   \Lab{\mu} )\\
\wnext \alwaysF ( \Lab{\psi}  \to   \Lab{\mu} )\\
\wnext \alwaysF ( \Lab{\mu}  \to \Lab{\varphi} \vee  \Lab{\psi} )\\
\Lab{\varphi} \to   \Lab{\mu} \\
\Lab{\psi}   \to    \Lab{\mu} \\
\Lab{\mu}   \to   \Lab{\varphi} \vee \Lab{\psi}
\end{array}$\\\hline
$\varphi \to \psi$ & $\alwaysF (\Lab{\mu} \leftrightarrow( \Lab{\varphi} \to \Lab{\psi}))$ & $\begin{array}{l}
	\wnext \alwaysF( \Lab{\mu}\wedge  \Lab{\varphi}  \to  \Lab{\psi})\\
	\wnext \alwaysF(\neg \Lab{\varphi}  \to   \Lab{\mu} )\\
	\wnext \alwaysF( \Lab{\psi}  \to   \Lab{\mu} )\\
	\wnext \alwaysF(\top  \to   \Lab{\varphi}\vee\neg \Lab{\psi} \vee \Lab{\mu} ) \\
	\Lab{\mu}\wedge \Lab{\varphi}   \to   \Lab{\psi} \\
	\neg \Lab{\varphi}  \to  \Lab{\mu}  \\
	\Lab{\psi}    \to   \Lab{\mu}    \\
	\Lab{\varphi}\vee \neg \Lab{\psi} \vee \Lab{\mu}\end{array}$\botline
\end{tabular}}

\end{table}

 \begin{table}[h!]\centering
\caption{Translation of the non-metric temporal operators.}
\label{tbl:tseitin:ltl}
{\tablefont
\begin{tabular}{@{\extracolsep{\fill}}ccc}\topline
$\mu$ & $\eta(\mu)$ & $\eta^*(\mu)$ \\\hline
$\previous \varphi$& $\begin{array}{c}
	\wnext \alwaysF (\Lab{\mu} \leftrightarrow \previous \Lab{\varphi})\\
	\neg \Lab\mu
\end{array}$ &$\begin{array}{c} 
	\wnext\alwaysF ( \Lab{\mu}  \to  \previous \Lab{\varphi} )\\
	\wnext\alwaysF ( \previous \Lab{\varphi}  \to   \Lab{\mu} )\\
	\neg \Lab{\mu}
\end{array}$\\\hline	
$\wprevious \varphi$& $\begin{array}{c}
	\wnext \alwaysF (\Lab{\mu} \leftrightarrow \previous \Lab{\varphi})\\
	\Lab\mu
\end{array}$ &$\begin{array}{c} 
	\wnext\alwaysF ( \Lab{\mu}  \to  \previous \Lab{\varphi} )\\
	\wnext\alwaysF ( \previous \Lab{\varphi}  \to   \Lab{\mu} )\\
	\Lab{\mu}
\end{array}$\\\hline	
$\next \varphi$& $\begin{array}{l}\wnext\alwaysF ( \previous \Lab{\mu} \leftrightarrow \Lab{\varphi})\\ \alwaysF (\finally \to \neg  \Lab{\mu}) \end{array}$ 
& $\begin{array}{l} 
	\wnext \alwaysF ( \previous \Lab{\mu} \to   \Lab{\varphi} )\\
	\wnext \alwaysF ( \Lab{\varphi}    \to  \previous \Lab{\mu} ) \\
	\alwaysF (\finally \to \neg  \Lab{\mu})
\end{array}$ \\\hline
$\wnext \varphi$& $\begin{array}{l}\wnext\alwaysF ( \previous \Lab{\mu} \leftrightarrow \Lab{\varphi})\\ \alwaysF (\finally \to  \Lab{\mu}) \end{array}$ 
& $\begin{array}{l} 
	\wnext \alwaysF ( \previous \Lab{\mu} \to   \Lab{\varphi} )\\
	\wnext \alwaysF ( \Lab{\varphi}    \to  \previous \Lab{\mu} ) \\
	\alwaysF (\finally \to  \Lab{\mu})
\end{array}$ \\\hline	
$ \varphi \since \psi$  & $\begin{array}{l} 
	\wnext \alwaysF (  \Lab{\mu} \leftrightarrow \Lab{\psi} \vee ( \Lab{\varphi} \wedge \previous \Lab{\mu}))\\
	\Lab\mu \leftrightarrow \Lab\psi
\end{array}$  & $\begin{array}{l} 
	\wnext \alwaysF (  \Lab{\mu} \to \Lab{\psi} \vee  \Lab{\varphi})\\
	\wnext \alwaysF (  \Lab{\mu} \to \Lab{\psi} \vee \previous \Lab{\mu} )\\
	\wnext \alwaysF(   \Lab{\psi} \to \Lab{\mu})\\
	\wnext\alwaysF ( \Lab{\varphi}\wedge \previous \Lab{\mu} \to \Lab{\mu}) \\
	\Lab{\mu}  \to \Lab{\psi}  \\
	\Lab{\psi} \to \Lab{\mu}  
\end{array}$\\\hline
$ \varphi \trigger \psi$ &$\begin{array}{l}\wnext \alwaysF (  \Lab{\mu} \leftrightarrow  \Lab{\psi} \wedge
	(\Lab{\varphi} \vee\previous \Lab{\mu})) \\ \Lab\mu \leftrightarrow \Lab\psi \end{array}$& 
$\begin{array}{l} 
	\wnext \alwaysF (  \Lab{\mu} \to \Lab{\psi} )\\
	\wnext \alwaysF (  \Lab{\mu} \to \Lab{\varphi} \vee\previous \Lab{\mu} )\\
	\wnext \alwaysF (  \Lab{\psi} \wedge \Lab{\varphi} \to \Lab{\mu}) \\
	\wnext \alwaysF (  \Lab{\psi} \wedge \previous \Lab{\mu} \to  \Lab{\mu})\\
	\Lab{\mu}   \to  \Lab{\psi}  \\
	\Lab{\psi}   \to \Lab{\mu}
\end{array}$\\\hline

$\varphi\until\psi$& $\begin{array}{l} \wnext\alwaysF( \previous \Lab{\mu} \leftrightarrow
	\previous \Lab{\psi} \vee (\previous \Lab{\varphi} \wedge \Lab{\mu}) ) \\
	\alwaysF (\finally \to ( \Lab{\mu} \leftrightarrow \Lab{\psi}) ) 
\end{array}$ & $\begin{array}{l}  
	\wnext \alwaysF( \previous \Lab{\mu} \to \previous \Lab{\psi} \vee \previous \Lab{\varphi}  )\\
	\wnext \alwaysF( \previous \Lab{\mu} \to \previous \Lab{\psi} \vee  \Lab{\mu} )\\
	\wnext \alwaysF ( \previous \Lab{\varphi} \wedge \Lab{\mu} \to \previous \Lab{\mu}  )\\
	\wnext \alwaysF( \previous \Lab{\psi}  \to  \previous \Lab{\mu}  )\\
	\alwaysF (\finally  \to ( \Lab{\mu}  \to  \Lab{\psi} ) )\\
	\alwaysF (\finally \to ( \Lab{\psi}  \to  \Lab{\mu} ) )
\end{array}$\\\hline
$\varphi\release\psi$ & $\begin{array}{l} 
	\wnext \alwaysF ( \previous \Lab{\mu} \leftrightarrow \previous \Lab{\psi} \wedge (\previous \Lab{\varphi} \vee \Lab{\mu}) )\\
	\alwaysF(\finally \to (\Lab{\mu} \leftrightarrow \Lab{\psi}))
\end{array}$& $\begin{array}{l} 
	\wnext \alwaysF ( \previous \Lab{\mu} \to \previous \Lab{\psi}  )\\
	\wnext \alwaysF ( \previous \Lab{\mu} \to \previous \Lab{\varphi} \vee \Lab{\mu})\\
	\wnext \alwaysF ( \previous \Lab{\psi} \wedge \Lab{\mu} \to \previous \Lab{\mu})\\
	\wnext\alwaysF ( \previous\Lab{\psi} \wedge  \previous\Lab{\varphi} \to  \previous\Lab{\mu} )\\
	\alwaysF (\finally \to ( \Lab{\mu} \to \Lab{\psi} ) )\\
	\alwaysF (\finally \to ( \Lab{\psi} \to \Lab{\mu} ) )\\						
\end{array}$\botline
\end{tabular}}	

\end{table}	
 
\clearpage{}

\end{document}